\documentclass[12pt,onecolumn]{IEEEtran}
\usepackage{color}
\usepackage{mdwlist}
\usepackage{subfigure}
\usepackage{setspace}
\usepackage{comment}
\usepackage{booktabs}
\usepackage{graphicx}
\usepackage{amsthm,amsfonts,amsmath,amssymb}
\usepackage[numbers]{natbib}
\usepackage[colorlinks=true,linkcolor=blue,citecolor=magenta]{hyperref}
\usepackage{multirow}
\usepackage{tabulary}
\usepackage{algorithm}
\usepackage{algorithmic}
\usepackage{enumerate}

\newcommand{\algrule}[1][1pt]{\par\vskip.5\baselineskip\hrule height #1\par\vskip.5\baselineskip}
\newcommand{\mbf}{\boldsymbol}

\newcommand{\nn}{\nonumber}

\newcommand{\pr}{Pr}

\newcommand{\complex}{\mathbb{C}}

\newcommand{\conj}[1]{{#1}^{\dagger}}

\newcommand{\expectation}{\mathbb{E}}

\newcommand{\floor}[1]{\lfloor{#1}\rfloor}

\newcommand{\noise}{\vec{z}}

\newcommand{\signal}{\vec{x}}
\newcommand{\signalc}[1]{x[{#1}]}
\newcommand{\transform}{\vec{X}}
\newcommand{\transformc}[1]{X[{#1}]}

\newcommand{\obs}{\vec{y}}

\newcommand{\binobsvi}[1]{\vec{y}_{b,#1}}

\newcommand{\binobsv}[2]{\vec{y}_{b,#1,#2}}
\newcommand{\binobsc}[2]{{y}_{b,#1,#2}}
\newcommand{\mlength}{L}

\newcommand{\stages}{d}
\newcommand{\setfactors}{{\cal F}}
\newcommand{\factor}[1]{f_{{#1}}}
\newcommand{\nbins}{n_b}
\newcommand{\delays}{D}

\newcommand{\osratio}{r}
\newcommand{\length}{n}
\newcommand{\sparsity}{k}
\newcommand{\sindex}{\delta}
\newcommand{\samples}{m}

\newcommand{\iterations}{\ell}
\newcommand{\binsize}{\mathbf{F}} 
\newcommand{\binexp}{\eta} 
\newcommand{\bbensemble}{{\cal C}^\sparsity_1(\setfactors,\nbins)}
\newcommand{\crtensemble}{{\cal C}^\sparsity_2(\setfactors,\length,\nbins)}

\newcommand{\primef}[1]{{\cal P}_{#1}}

\newtheorem{theorem}{Theorem}[section]

\newtheorem{remark}[theorem]{Remark}

\newtheorem{definition}[theorem]{Definition}

\newtheorem{lemma}[theorem]{Lemma}

\newtheorem{corollary}[theorem]{Corollary}

\begin{document}
\title{A FFAST framework for computing a $\sparsity$-sparse DFT in $O(\sparsity \log\sparsity)$ time using sparse-graph alias codes}
\author{
\authorblockN{Sameer~Pawar~and~Kannan~Ramchandran\\}
\authorblockA{Dept. of Electrical Engineering and Computer Sciences \\
University of California, Berkeley \\
\{spawar, kannanr\}@eecs.berkeley.edu}
\thanks{The material in this paper was presented in part at the IEEE International Symposium on Information Theory, Istanbul, Turkey, 2013.} 
}

\maketitle
\begin{abstract}\label{sec:abstract} Given an $\length$-length input signal $\signal$, it is well known that its Discrete Fourier Transform (DFT), $\transform$, can be computed from $\length$ samples in $O(\length \log \length)$ operations using a Fast Fourier Transform (FFT) algorithm. If the spectrum $\transform$ is $\sparsity$-sparse (where $\sparsity<<\length$), can we do better?  We show that asymptotically in $\sparsity$ and $\length$, when $\sparsity$ is sub-linear in $\length$ (precisely, $\sparsity = O(\length^{\sindex}$) where $0 < \sindex <1$), and the support of the non-zero DFT coefficients is uniformly random, our proposed FFAST (Fast Fourier Aliasing-based Sparse Transform) algorithm computes the DFT $\transform$, from $O(\sparsity)$ samples in $O(\sparsity\log\sparsity)$ arithmetic operations, with high probability. Further, the constants in the big Oh notation for both sample and computational cost are small, e.g., when $\sindex < 0.99$, which essentially covers almost all practical cases of interest, the sample cost is less than $4\sparsity$.

Our approach is based on filterless subsampling of the input signal $\signal$ using a set of {\em carefully chosen uniform subsampling patterns guided by the Chinese Remainder Theorem (CRT)}. The idea is to cleverly exploit, rather than avoid,  the resulting aliasing artifacts induced by subsampling.  Specifically, the subsampling operation on $\signal$ is designed to create aliasing patterns on the spectrum $\transform$ that ``look like" parity-check constraints of a good erasure-correcting sparse-graph code.  {\em Next, we show that computing the sparse DFT $\transform$ is equivalent to decoding of sparse-graph codes}. These codes further allow for fast peeling-style decoding. The resulting DFT computation is low in both sample complexity and decoding complexity. We analytically connect our proposed CRT-based aliasing framework to random sparse-graph codes, and analyze the performance of our algorithm using density evolution techniques from coding theory. We also provide simulation results, that are in tight agreement with our theoretical findings.


\end{abstract}

\section{Introduction}\label{sec:intro}
Spectral analysis using the Discrete Fourier Transform (DFT) has been of universal importance in engineering and scientific applications for a long time.  The Fast Fourier Transform (FFT) is the fastest known way to compute the DFT of an arbitrary $\length$-length signal, and has a computational complexity of $O(\length \log \length)$\footnote{Recall that a single variable function $f(x)$ is said to be $O(g(x))$, if for a sufficiently large $x$ the function $|f(x)|$ is bounded above by $|g(x)|$, i.e., $\lim_{x \rightarrow \infty}|f(x)| < c |g(x)|$ for some constant $c$. Similarly, $f(x) = \Omega(g(x))$  if $\lim_{x \rightarrow \infty}|f(x)| > c |g(x)|$ and $f(x) = o(g(x))$ if the growth rate of $|f(x)|$ as $x\rightarrow \infty$, is negligible as compared to that of $|g(x)|$, i.e. $\lim_{x\rightarrow \infty}|f(x)|/|g(x)| = 0$.}. Many applications of interest involve signals, e.g. audio, image, video data, biomedical signals etc., which have a sparse Fourier spectrum. In such cases, a small subset of the spectral components typically contain most or all of the signal energy, with most spectral components being either zero or negligibly small.  If the $\length$-length DFT, $\transform$, is $\sparsity$-sparse, where $\sparsity<<\length$, can we do better in terms of both {\bf sample} and {\bf computational} complexity of computing the sparse DFT? We answer this question affirmatively. 
In particular, we show that asymptotically in $\sparsity$ and $\length$, when $\sparsity$ is sub-linear in $\length$ (precisely, $\sparsity = O(\length^{\sindex}$) where $0 < \sindex <1$), and the support of the non-zero DFT coefficients is uniformly random, our proposed FFAST (Fast Fourier Aliasing-based Sparse Transform) algorithm computes the DFT $\transform$, from judiciously chosen $O(\sparsity)$ samples in $O(\sparsity\log\sparsity)$ arithmetic operations, with high probability. Further, the constants in the big Oh notation for both sample and computational cost are small, e.g., when $\sindex < 0.99$, which essentially covers almost all practical cases of interest, the sample cost is less than $4\sparsity$.

At a high level, our idea is to cleverly exploit rather than avoid the aliasing resulting from subsampling, and to shape the induced spectral artifacts to look like parity constraints of ``good" erasure-correcting codes, e.g., Low-Density-Parity-Check (LDPC) codes \cite{gallager1962low} and fountain codes \cite{luby2002digital}, that have both low computational complexity and are near-capacity achieving. Towards our goal of shaping the spectral artifacts to look like parity constraints of good erasure-correcting codes, we are challenged by not being able to design any arbitrary spectral code at will, as we can control only the subsampling operation on the time-domain signal. The key idea is to design subsampling patterns, guided by the {\em Chinese-Remainder-Theorem (CRT)} \cite{blahut1985fast}, that create the desired code-like aliasing patterns. Based on the qualitative nature of the subsampling patterns needed, our analysis is decomposed into two regimes (see Section~\ref{sec:verysparse} and Section~\ref{sec:lesssparse} for more details): 
\begin{itemize}
\item The ``very-sparse" regime, where $\sparsity = O(\length^{\sindex}), 0 < \sindex \leq 1/3$. For the very sparse regime the subsampling patterns are based on relatively co-prime numbers.
\item The ``less-sparse" regime, where $\sparsity = O(\length^{\sindex}), 1/3 < \sindex < 1$. The sub-sampling pattern, for the less-sparse regime, comprise of ``cyclically-shifted" overlapping co-prime integers.
\end{itemize}

\begin{figure}[h]
\begin{center}
\includegraphics[width = 0.8\linewidth]{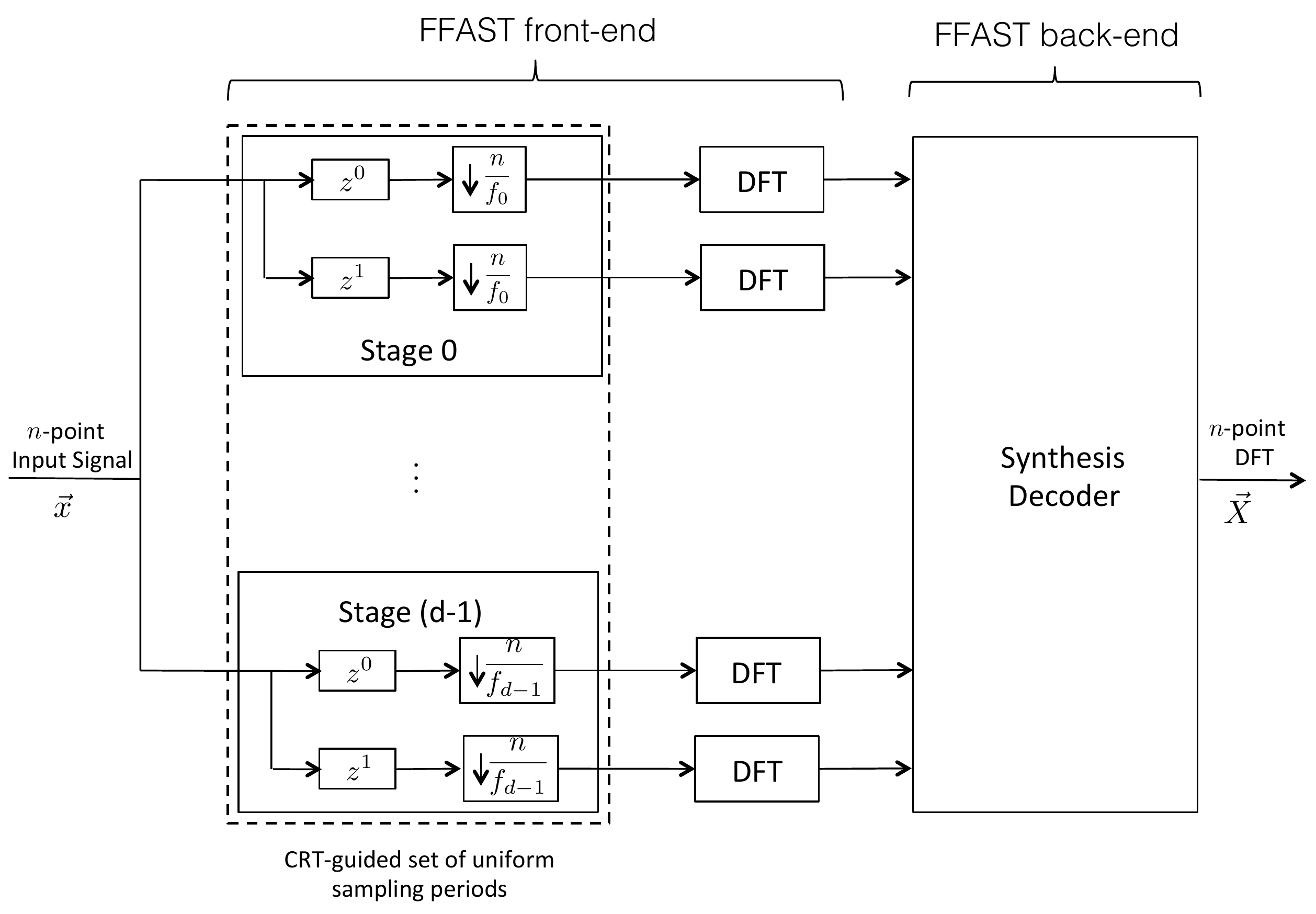}
\caption{Block diagram of the FFAST architecture. The $\length$-length input signal $\signal$ is uniformly subsampled, through $\stages$ stages. Each stage subsamples the input signal and its circularly shifted version by a carefully chosen (guided by the CRT as explained in Section~\ref{sec:verysparse} and Section~\ref{sec:lesssparse}) set of sampling periods, to generate $O(\sparsity)$ samples per sub-sampling path. Next, the big $\length$-length DFT $\transform$ is synthesized, from the short (length $O(\sparsity))$ DFTs, using the FFAST back-end decoder.}
\label{fig:conceptual}
\end{center}
\end{figure}

Our approach is summarized in Fig.~\ref{fig:conceptual}. The $\length$-length input signal $\signal$ is uniformly subsampled, through a small number\footnote{We show that the number of stages depend on the sparsity index $\sindex$, and is in the range of $3$ to $8$ for $\sindex \leq 0.99$, as shown in Table~\ref{tab:const}.} of stages, say $\stages$. Each stage subsamples the input signal and its circularly shifted version by a CRT guided set (see Sections \ref{sec:verysparse} and \ref{sec:lesssparse}) of sampling periods, to generate $\factor{i} = O(\sparsity)$ samples per sub-sampling path, for $i = 0,\hdots,\stages-1$. Next, the large $\length$-length DFT $\transform$ is synthesized from much smaller (length $O(\sparsity))$ DFTs, using the FFAST back-end peeling decoder. For the entire range of practical interest of sub-linear sparsity, i.e., $0 < \sindex < 0.99$, the FFAST algorithm computes $\length$-length $\sparsity$-sparse DFT $\transform$ in $O(\sparsity \log \sparsity)$ computations from less than $3.74\sparsity$ samples. It is gratifying to note that both the sample complexity and the computational complexity depend only on the sparsity parameter $\sparsity$, which is sub-linear in $\length$.

We emphasize the following caveats.  First, our analytical results are probabilistic, and hold for asymptotic values of $\sparsity$ and $\length$, with a success probability that approaches $1$ asymptotically. This contrasts the $O(\length \log \length)$ FFT algorithm which works deterministically for all values of $\sparsity$ and $\length$. Secondly, we assume a uniformly random model for the support of the non-zero DFT coefficients. Lastly, we require the signal length $\length$ to be a product of a few (typically 3 to 9) distinct primes\footnote{This is not a major restriction for two reasons. Firstly, in many problems of interest, the choice of $\length$ is available to the system designer, and choosing $\length$ to be a power of 2 is often invoked only to take advantage of the readily-available radix-2 FFT algorithms. Secondly, by truncating or zero-padding the given signal, by a constant number of samples, one can obtain a modified signal of a desired length $\length$. The DFT of the modified signal has dominant non-zero DFT coefficients at the same frequencies as the original signal but with additional noise-like side-lobes, due to windowing effect, and can be decoded using a noise robust version of the FFAST algorithm \cite{pawar2014isit}.}. In effect, our algorithm trades off sample and computational complexity for asymptotically zero probability of failure guarantees, for a uniformly random model of the support. The FFAST algorithm is applicable whenever $\sparsity$ is sub-linear in $\length$ (i.e. $\sparsity$ is $o(\length)$), but is obviously most attractive when $\sparsity$ is much smaller than $\length$. As a concrete example, when $\sparsity=300$, and $\length= 2^73^55^3\approx 3.8\times10^6$, our algorithm achieves computational savings by a factor of more than $6000$, and savings in the number of input samples by a factor of more than $3900$ over the standard FFT (see \cite{temperton1983self} for computational complexity of a prime factor FFT algorithm). This can be a significant advantage in many existing applications, as well as enable new classes of applications not thought practical so far.

The focus of this paper is on signals having an exactly-sparse DFT. Our motivation for this focused study is three-fold: (i) to provide conceptual clarity of our proposed approach in a noiseless setting; (ii) to present our deterministic subsampling front-end measurement subsystem as a viable alternative to the class of randomized measurement matrices popular in the compressive sensing literature \cite{donoho2006compressed, candes2006near}; and (iii) to explore the fundamental limits on both sample complexity and computational complexity for the exactly sparse DFT problem, which is of intellectual interest.
In addition, the key insights derived from analyzing the exactly sparse signal model, apply more broadly to the noisy setting (see discussion in Section~\ref{sec:conclusion}), i.e., where the observations are further corrupted by noise \cite{pawar2014isit}.

The rest of the paper is organized as follows. Section~\ref{sec:problem} states the problem and introduce important notations. Section~\ref{sec:mainresults} presents our main results, and the related literature is reviewed in Section~\ref{sec:relatedwork}. Section~\ref{sec:DFTSparseCodes} exemplifies the mapping of the problem of computing a sparse DFT to decoding over an appropriate sparse-graph code. In Section~\ref{sec:verysparse} and Section~\ref{sec:lesssparse} we analyze the performance of the FFAST algorithm for the {\em very-sparse} and the {\em less-sparse} regimes respectively. Section~\ref{sec:experiments} provides simulation results that empirically corroborate our theoretical findings. The paper is concluded in Section~\ref{sec:conclusion}, with few remarks about extending the FFAST framework for noise robustness and other applications.
\section{Problem formulation, notations and preliminaries}\label{sec:problem} 
\subsection{Problem formulation} Consider an $\length$-length discrete-time signal $\signal$ that is sum of $\sparsity << \length$ complex exponentials, i.e., its $\length$-length discrete Fourier transform has $\sparsity$ non-zero coefficients:
\begin{equation}
\signalc{p} = \sum_{q = 0}^{\sparsity-1} \transformc{\ell_q}e^{2\pi \imath\ell_q p/\length}, \ \ \ \ p = 0,1,\hdots,\length-1,
\end{equation}
where the discrete frequencies $\ell_q \in \{0,1,\hdots,\length-1 \}$ and the amplitudes $ \transformc{\ell_q} \in {\mathbb{C}}$, for $q=0,1,\hdots,\sparsity-1$. We further assume that the discrete frequencies $\ell_q$ are uniformly random in $0$ to $\length-1$ and the amplitudes $\transformc{\ell_q}$ are drawn from some continuous distribution over complex numbers. Under these assumptions, we consider the problem of identifying the $\sparsity$ unknown frequencies $\ell_q$ and the corresponding complex amplitudes $\transformc{\ell_q}$ from the time domain samples $\signal$. A straightforward solution is to compute an $\length$-length DFT, $\transform$, using a standard FFT algorithm \cite{blahut1985fast}, and locate the $\sparsity$ non-zero coefficients. Such an algorithm requires $\length$ samples and $O(\length\log \length)$ computations. When the DFT $\transform$ is known to be exactly $\sparsity$-sparse and $\sparsity << \length$, computing all the $\length$ DFT coefficients seems redundant.

In this paper, we address the problem of designing an algorithm, to compute the $\sparsity$-sparse $\length$-length DFT of $\signal$ for the asymptotic regime of $\sparsity$ and $\length$. We would like the algorithm to have the following features:
\begin{itemize}
\item it takes as few input samples $\samples$ of $\signal$ as possible.
\item it has a low computational cost, that is possibly a function of only the number of input samples $\samples$.
\item it is applicable for the entire sub-linear regime, i.e., for all $0 < \sindex < 1$, where $\sparsity =O(\length^{\sindex})$.
\item it computes the DFT $\transform$ with a probability of failure vanishing to $0$ as $\samples$ becomes large.
\end{itemize}

In Table~\ref{tab:glossary}, we provide notations and definitions of the important parameters used in the rest of the paper.

\begin{table}[h]
\begin{center}
\begin{tabular}{|c|c|}
  \hline
  Notation & Description\\
  \hline
  $\length$ & Signal length. \\  
  \hline
  $\sparsity$ & Number of non-zero coefficients in the DFT $\transform$.\\ 
  \hline
  $\sindex$& Sparsity-index: $\sparsity =O(\length^{\sindex}), 0 < \sindex < 1$.\\
  \hline
  $\samples$& Sample complexity: Number of samples of $\signal$ input to the FFAST algorithm.\\
  \hline
  $\osratio = \samples/\sparsity$ & Oversampling ratio: Input samples per non-zero DFT coefficient of the signal.\\
  \hline
  $\stages$& Number of stages in the sub-sampling ``front-end" architecture of the FFAST.\\
  \hline
  {$\factor{i}$} & Number of output samples at each sub-sampling path of stage $i$ of the FFAST front-end.\\
   \hline
\end{tabular}
\end{center}
\caption{Glossary of important notations and definitions used in the rest of the paper. The last two parameters are related to the FFAST ``front end" architecture (see Figure~\ref{fig:conceptual} for reference).}\label{tab:glossary}
\end{table}

 \subsection{The Chinese-Remainder-Theorem (CRT)} The CRT plays an important role in our proposed FFAST architecture as well as in the FFAST decoder. For integers $a,N$, we use $(a)_N$  to denote the operation, $a \mod N$, i.e., $(a)_N\triangleq a \mod N$.

\begin{theorem}[Chinese-Remainder-Theorem \cite{blahut1985fast}]\label{thm:CRT}
Suppose $n_0,n_1,\hdots,n_{\stages-1}$ are pairwise co-prime positive integers and $N = \prod_{i=0}^{\stages-1}n_i$. Then, every integer `$a$' between $0$ and $N-1$ is uniquely represented by the sequence $r_0,r_1,\hdots,r_{\stages-1}$ of its remainders modulo $n_0,\hdots,n_{\stages-1}$ respectively and vice-versa. 
\end{theorem}
Further, given a sequence of remainders $r_0,r_1,\hdots,r_{\stages-1}$, where $0 \leq r_i < n_i$, one can find an integer `$a$', such that,
\begin{equation}
(a)_{n_i} \equiv r_i \ \text{ for } i = 0,1,\hdots,\stages-1.
\end{equation}

For example, consider the following pairwise co-prime integers $n_0 = 3, n_1 = 4$ and $n_2 = 5$. Then, given a sequence of remainders $r_0 = 2, r_1=2, r_2=3$, there exists a unique integer `$a$', such that,
\begin{eqnarray}\label{eq:crtsystem}
2 &\equiv& a \mod 3\nonumber\\
2 &\equiv& a \mod 4\\
3 &\equiv& a \mod 5\nonumber
\end{eqnarray}
It is easy to verify that $a=38$ is a unique integer, between $0$ and $59$, that satisfies the congruencies in \eqref{eq:crtsystem}.
 \section{Main Results}\label{sec:mainresults} 
\begin{figure}[t]
\begin{center}
\includegraphics[width = 0.7\linewidth]{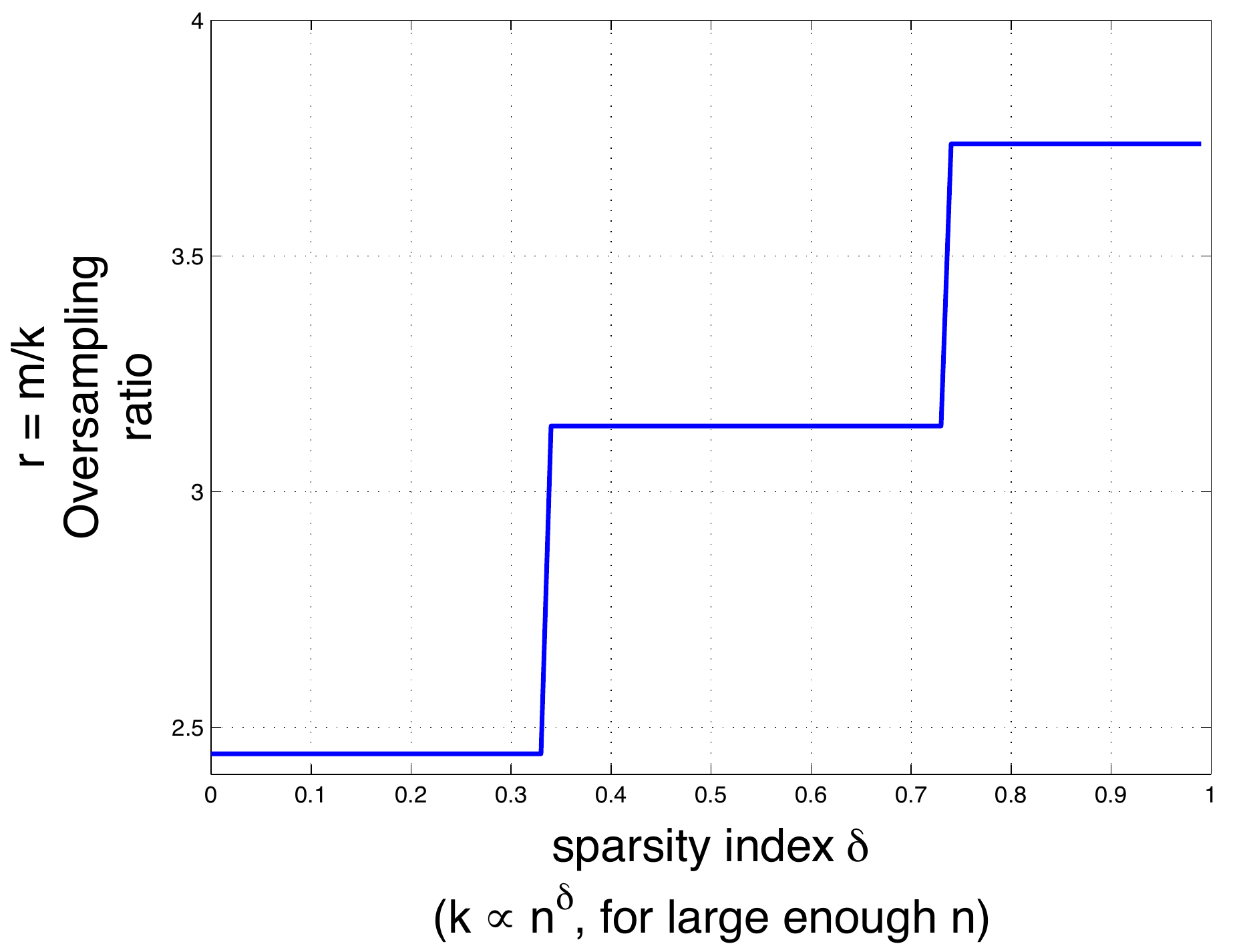}
\caption{The plot shows the relation between the oversampling ratio $\osratio = \samples/\sparsity$, and the sparsity index $\sindex$, for $0 < \sindex < 0.99$, where $\sparsity= O(\length^{\sindex})$. The FFAST algorithm successfully computes the $\sparsity$-sparse $\length$-length DFT $\transform$ of the desired $\length$-length signal $\signal$ with high probability, as long as the number of samples $\samples$ is above the threshold given in the plot. Note that for nearly the entire sub-linear regime of practical interest, e.g. $\sparsity < \length^{0.99}$, the oversampling ratio $\osratio < 4$. The above plot is an achievable performance of the FFAST algorithm using the constructions described in Section~\ref{sec:verysparseconstruction} and Section~\ref{sec:lesssparseconstruction}. There are many other constructions of FFAST architecture that may achieve similar or better performance for a specific sparsity index $\sindex$. }
\label{fig:constants}
\end{center}
\end{figure}  

We propose a novel FFAST architecture and algorithm to compute a $\sparsity$-sparse $\length$-length DFT, $\transform$, of an $\length$-length signal $\signal$. In this paper, an $\length$-length input signal $\signal$ is said to have a $\sparsity$-sparse DFT $\transform$, if $\transform$ has {\em at most} $\sparsity$ non-zero coefficients, whose locations are uniformly randomly distributed in $\{0,1,\ldots,\length-1\}$. The FFAST algorithm computes the $\sparsity$-sparse $\length$-length DFT with high probability, using as few as $O(\sparsity)$ samples of $\signal$ and $O(\sparsity \log \sparsity)$ computations. The following theorem states the main result of the paper.

\begin{theorem}\label{thm:main}
For any $0 < \sindex < 1$, and large enough $\length$, the FFAST algorithm computes a $\sparsity$-sparse DFT $\transform$ of an $\length$-length input $\signal$, where $\sparsity =O(\length^\sindex)$, with the following properties:
\begin{enumerate}
\item {\bf Sample complexity:} The algorithm needs $\samples = \osratio\sparsity$ samples of $\signal$, where $\osratio>1$ is a small constant and depends on the sparsity index $\sindex$;
\item {\bf Computational complexity:}  The FFAST algorithm computes DFT $\transform$ using $O(\sparsity\log\sparsity)$, arithmetic operations.
\item {\bf Probability of success:} The probability that the algorithm correctly computes the $\sparsity$-sparse DFT $\transform$ is at least 1- $O(1/\samples)$.
\end{enumerate}
\end{theorem}
\begin{proof}
We prove the theorem in three parts. In Section~\ref{sec:verysparse}, we analyze the performance of the FFAST algorithm for the very-sparse regime of $0 < \sindex \leq 1/3$, and in Section~\ref{sec:lesssparse} we analyze the less-sparse regime of $1/3 < \sindex < 1$. Lastly, in Section~\ref{sec:complexity} we analyze the sample and computational complexity of the FFAST algorithm.\end{proof}

\begin{remark}\label{rem:delta}[Oversampling ratio $\osratio$]
{The achievable oversampling ratio $\osratio$ in the sample complexity $\samples = rk$}, depends on the number of stages $\stages$ used in the FFAST architecture. The number of stages $\stages$, in turn, is a function of the sparsity index $\sindex$ (recall $\sparsity = O(\length^{\sindex}$)), and increases as $\sindex \rightarrow 1$ (i.e., as the number of the non-zero coefficients, $\sparsity$, approach the linear regime in $\length$). In Sections \ref{sec:verysparse} and \ref{sec:lesssparse}, we provide constructions of the FFAST front-end architecture, that require increasing number of stages $\stages$ as $\sindex$ increase from $0$ to $1$. In our proposed construction, the increase in $\stages$ occurs over intervals of $\sindex$, resulting in a staircase plot as shown in Fig.~\ref{fig:constants}. Note, that the plot in Fig.~\ref{fig:constants} is an achievable performance of the FFAST algorithm using the constructions described in Section~\ref{sec:verysparseconstruction} and Section~\ref{sec:lesssparseconstruction}. There are many other constructions of FFAST architecture that may achieve similar or better performance for a specific sparsity index $\sindex$. Table~\ref{tab:const} provides some example values of $\osratio$ and $\stages$ for different values of the sparsity index $\sindex$.
\begin{table}[h]
\begin{center}
\begin{tabular}{|c|c|c|c|c|c|}
  \hline
  $\sindex$ & 1/3 & 2/3 & 0.99 & 0.999 & 0.9999\\
  \hline
  \stages & 3 &  6   & 8 &  11  & 14 \\
  \hline
  \osratio & 2.45 &  3.14  & 3.74 &  4.64  & 5.51 \\
  \hline
\end{tabular}
\end{center}
\caption{The table shows the number of subsampling stages $\stages$ used in the FFAST architecture, and the corresponding values of the oversampling ratio $\osratio$, for different values of the sparsity index $\sindex$.}\label{tab:const}
\end{table}

\end{remark}

\section{Related work}\label{sec:relatedwork} The problem of computing a sparse discrete Fourier transform of a signal is related to the rich literature of frequency estimation \cite{prony1795essai, pisarenko1973retrieval, schmidt1986multiple, roy1989esprit} in statistical signal processing as well as compressive-sensing \cite{donoho2006compressed, candes2006near}. In frequency estimation, it is assumed that a signal consists of $\sparsity$ complex exponentials embedded in additive noise. The frequency estimation techniques are based on well-studied statistical methods like MUSIC and ESPRIT \cite{prony1795essai, pisarenko1973retrieval, schmidt1986multiple, roy1989esprit}, where the focus is on `super-resolution' spectral estimation from initial few consecutive samples, i.e., extrapolation. In contrast, the algorithm proposed in this paper combine tools from coding theory, number theory, graph theory and signal processing, to `interpolate' the signal from interspersed but significantly less number of samples. In compressive sensing, the objective is to reliably reconstruct the sparse signal from as few measurements as possible, using a fast recovery technique. The bulk of this literature concentrates on random linear measurements, followed by either convex programming or greedy pursuit reconstruction algorithms \cite{candes2006near, candes2006robust, tropp2007signal}. An alternative approach, in the context of sampling a continuous time signal with a finite rate of innovation is explored in \cite{vetterli2002sampling, dragotti2007sampling, blu2008sparse, mishali2010theory}. Unlike the compressive sensing problem, where the resources to be optimized are the number of measurements\footnote{Consider a compressive sensing problem with a measurement matrix $A$, i.e., $\mbf{y} = A\signal$, where $\mbf{y}$ is a measurement vector and $\signal$ is the input signal. Then, the sample complexity is equal to the number of non-zero columns of $A$ and the measurement complexity is equal to the number of non-zero rows of $A$.} and the recovery cost, in our problem, we want to minimize the {\em number of input samples} needed by the algorithm in addition to the recovery cost. 

At a higher level though, despite some key differences in our approach to the problem of computing a sparse DFT, our problem is indeed closely related to the spectral estimation and compressive sensing literature, and our approach is naturally inspired by this, and draws from the rich set of tools offered by this literature.   

A number of previous works \cite{GGI02, GMS05, gilbert2008tutorial, hassanieh2012nearly, hassanieh2012simple} have addressed the problem of computing a $1$-D DFT of a discrete-time signal that has a sparse Fourier spectrum, in sub-linear sample and time complexity. Most of these algorithms achieve a sub-linear time performance by first isolating the non-zero DFT coefficients into different bins, using specific filters or windows that have `good' (concentrated) support in both, time and frequency. The non-zero DFT coefficients are then recovered iteratively, one at a time. The filters or windows used for the binning operation are typically of length $O(\sparsity \log (\length))$. As a result, the sample and computational complexity is typically $O(\sparsity \log (\length))$ or more. Moreover the constants involved in the big-Oh notation can be large, e.g., the empirical evaluation of \cite{GMS05} presented in \cite{iwen2007empirical} shows that for $\length=2^{22}$ and $\sparsity=7000$, the number of samples required are $\samples \approx 2^{21} = 300k$ which is $75$ times more than the sample complexity $4k$ of the FFAST algorithm\footnote{As mentioned earlier, the FFAST algorithm requires the length of the signal $\length$ to be a product of a few distinct primes. Hence, the comparison is for an equivalent $\length \approx 2^{22}$ and $\sparsity=7000$.}. The work of \cite{gilbert2008tutorial} provides an excellent tutorial on some of the key ideas used by most sub-linear time sparse FFT algorithms in the literature. While we were writing this paper, we became aware of a recent work \cite{ghazi2013sample}, in which the authors consider the problem of computing a noisy as well as an exactly-sparse $2$-D DFT of size $\sqrt{\length} \times \sqrt{\length}$ signal. For an exactly sparse signal, and when $\sparsity = O(\sqrt{\length})$, the algorithm in \cite{ghazi2013sample} uses $O(\sparsity)$ samples to compute the 2-D DFT of the signal in $O(\sparsity\log\sparsity)$ time with a constant probability of failure (that is controllable but that does not appear go to zero asymptotically). In \cite{iwen2010combinatorial}, the author proposes a sub-linear time algorithm with a sample complexity of $O(\sparsity \log^4n)$ or $O(\sparsity^2\log^4n)$ and computational complexity of $O(\sparsity\log^5n)$ or $O(\sparsity^2\log^4n)$ to compute a sparse DFT, with high probability or zero-error respectively. The algorithm in \cite{iwen2010combinatorial} exploits the Chinese-Remainder-Theorem, along with $O(poly(\log \length))$ number of subsampling patterns to identify the locations of the non-zero DFT coefficients. In contrast, the FFAST algorithm exploits the CRT to induce `good' sparse-graph codes using a small constant number of subsampling patterns and computes the sparse DFT with a vanishing probability of failure.

\section{Computing DFT using decoding of alias-codes}\label{sec:DFTSparseCodes} 
\begin{figure}[h]
\begin{center}
\includegraphics[width=\linewidth]{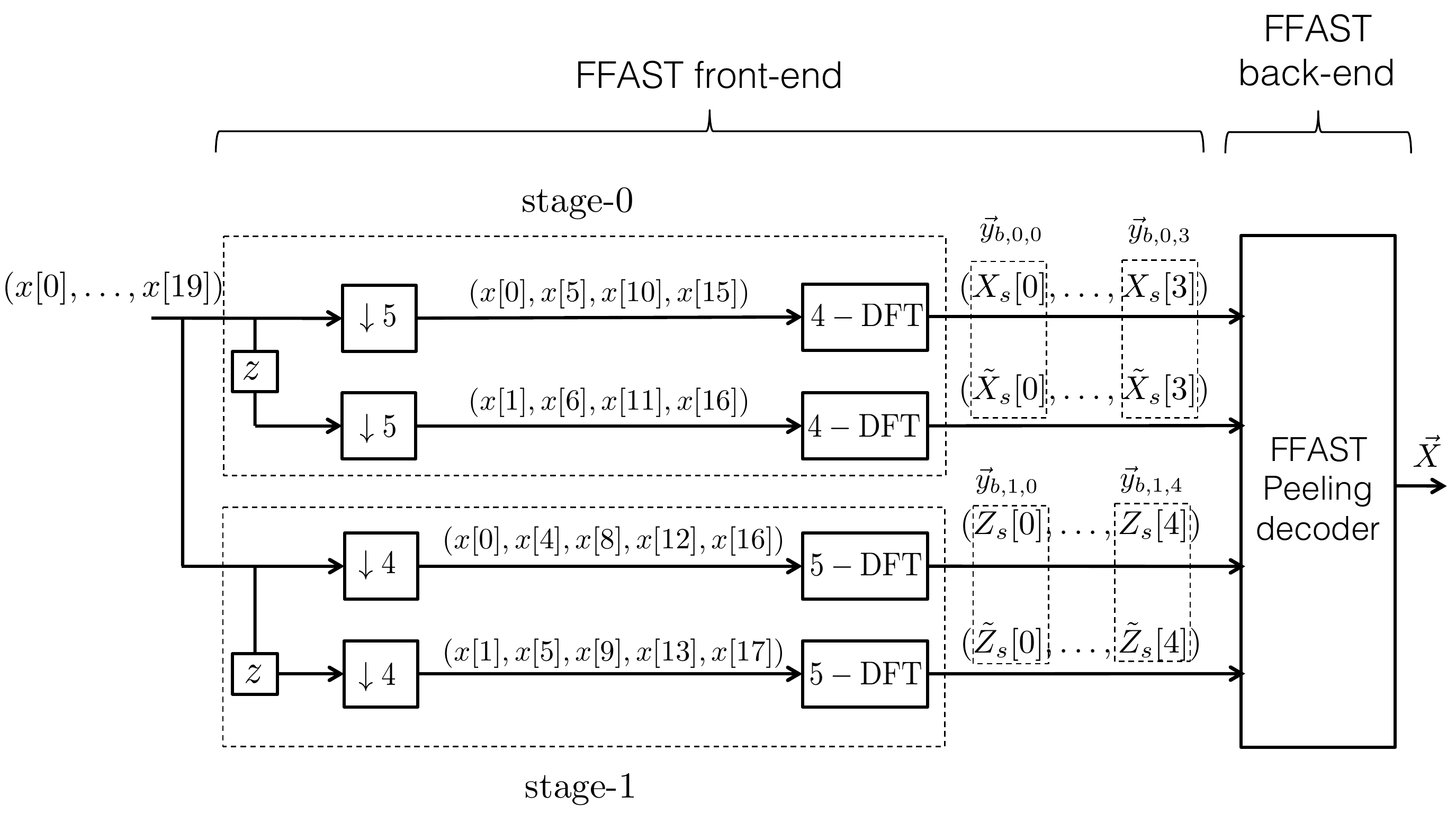}
\caption{An example FFAST architecture. The input to the FFAST architecture is a $20$-length discrete-time signal $\signal = (\signalc{0},\hdots,\signalc{19})$. The input signal and its circularly shifted version are first subsampled by $5$ to obtain two streams of sampled signal, also referred as delay-chains, each of length $\factor{0} = 4$. A $4$-length DFT of the output of each dely-chain is then computed to obtain the observations $(X_s[.],\tilde{X}_s[.])$. Similarly, downsampling by $4$ followed by a $5$-length DFT provides the second set of $\factor{1} = 5$ observations $(Z_s[.],\tilde{Z}_s[.])$. Note that the number of output samples $\factor{0}$ and $\factor{1}$ in the two different stages  are pairwise co-prime and are factors of $\length=20$. In general, the number of stages and the choice of the subsampling factors depend on the sparsity index $\sindex$, as will be described in Section~\ref{sec:verysparse} and Section~\ref{sec:lesssparse}.}
\label{fig:ffastex}
\end{center}
\end{figure}

In this section, we use a simple example to illustrate the working mechanics of the FFAST sub-sampling ``front-end" and the associated ``back-end" peeling-decoder. Then we demonstrate, how the output of the FFAST front-end sub-sampling can be viewed as a sparse graph code in the frequency domain, which we refer to as ``Alias-codes", and computing the DFT is equivalent to decoding of this resulting alias-code. Later, we also point out a connection between the FFAST and coding techniques for a packet erasure channel.

Consider a $20$-length discrete-time signal $\signal = (\signalc{0},\hdots,\signalc{19})$, such that its $20$-length DFT $\transform$, is $5$-sparse. Let the $5$ non-zero DFT coefficients of the signal $\signal$ be $\transformc{1} = 1, \transformc{3}=4, \transformc{5} = 2, \transformc{10} = 3$ and $\transformc{13}=7$. 

\subsection{FFAST sub-sampling front-end}\label{sec:frontEnd} In general, the FFAST sub-sampling front-end architecture, as shown in Fig.~\ref{fig:conceptual}, consists of multiple sub-sampling stages $\stages$ and each stage further has $2$ sub-sampling paths, referred to as ``delay-chains". The sampling periods used in both the delay-chains of a stage are identical. For example, consider a $2$ stage FFAST sub-sampling front-end, shown in Fig.~\ref{fig:ffastex}, that samples the $20$-length input signal $\signal$ and its circularly shifted version\footnote{Conventionally, in signal processing literature $z$ is used to denote a time-delay. In this paper, for ease of exposition we use $z$ to denote a time-advancement (see Fig.~\ref{fig:ffastex} for an example).} by factors $5$ and $4$ respectively. In the sequel we use, $\factor{i}$ to denote the number of the output samples per delay-chain of stage $i$, e.g., $\factor{0}=4$ and $\factor{1} = 5$ in Fig.~\ref{fig:ffastex}. The FFAST front-end then computes short DFT's of  appropriate lengths of the individual output data streams of the delay-chains. Next, we group the output of short DFT's into ``bin-observation".
\subsubsection{\textbf{Bin observation}}\label{sec:binobs} A bin-observation is a $2$-dimensional vector formed by collecting one scalar output value from the DFT output of the signal from each of the $2$ delay chains in a stage. For example, $\binobsv{0}{1}$ is an observation vector of bin $1$ in stage $0$ and is given by,
\begin{equation}\label{eq:bin0}
\binobsv{0}{1} = 
\left(
\begin{array}{c}
  X_s[1]\\
  \tilde{X}_s[1]
  \end{array}
\right).
\end{equation}
The first index of the observation vector corresponds to the stage number, while the second index is the bin number within a stage. Note that in the FFAST architecture of Fig.~\ref{fig:ffastex}, there are total of $4$ bins in stage $0$ and $5$ bins in stage $1$. 

Using basic Fourier transform properties, reviewed below, of sampling, aliasing and circular shift, one can compute the relation between the original $20$-length DFT $\transform$ of the signal $\signal$ and the output bin-observations $\{\binobsv{i}{j}\}$ of the FFAST front-end.
\begin{itemize}
\item {\bf Aliasing:} If a signal is subsampled in the time-domain, its frequency components mix together, i.e., alias, in a pattern that depends on the sampling procedure. For example, consider the output of the first delay-chain of the stage-0 in Fig.~\ref{fig:ffastex}. The input signal $\signal$ is uniformly sampled by a factor of $5$ to get $\signal_s = (\signalc{0}, \signalc{5}, \signalc{10}, \signalc{15})$. Then, the $4$-length DFT of $\signal_s$ is related to the original $20$-length DFT $\transform$ as:
\begin{eqnarray*}
X_s[0] &=& \transformc{0} + \transformc{4} + \transformc{8} + \transformc{12} + \transformc{16} = 0\\
X_s[1] &=& \transformc{1} + \transformc{5} + \transformc{9} + \transformc{13} + \transformc{17} = 10\\
X_s[2] &=& \transformc{2} + \transformc{6} + \transformc{10} + \transformc{14} + \transformc{18} = 3\\
X_s[3] &=& \transformc{3} + \transformc{7} + \transformc{11} + \transformc{15} + \transformc{19} =4
\end{eqnarray*}

More generally, if the sampling period is $N$ (we assume that $N$ divides $\length$) then,
\begin{equation}\label{eq:aliasing_pattern}
X_s[i] = \sum_{j \equiv (i)_{\length/N}} \transformc{j},
\end{equation}
where the notation $j \equiv (i)_{\length/N}$, denotes $j \equiv i \ \text{mod} \ \length/N$.

\item {\bf Circular Shift in time:} A circular shift in the time-domain results in a phase shift in the frequency-domain. Consider a circularly shifted signal $\signal^{(1)}$ obtained from $\signal$ as $x^{(1)}[{i}] = \signalc{(i+1)_{\length}}$. The DFT coefficients of the shifted signal $\signal^{(1)}$, are given as, $X^{(1)}[j] = \omega^{j}_n\transformc{j}$, where $\omega_n = \exp(2\pi\imath/\length)$ is an $\length^{th}$ root of unity. In general a circular shift of $\length_0$ results in $X^{(n_0)}[j] = \omega^{jn_0}_n\transformc{j}$. 
\end{itemize}

\subsection{Alias-codes and its connection to computing sparse DFT from sub-sampled signal}\label{sec:alias codes}
\begin{figure}[h]
\begin{center}
\includegraphics[scale=.5]{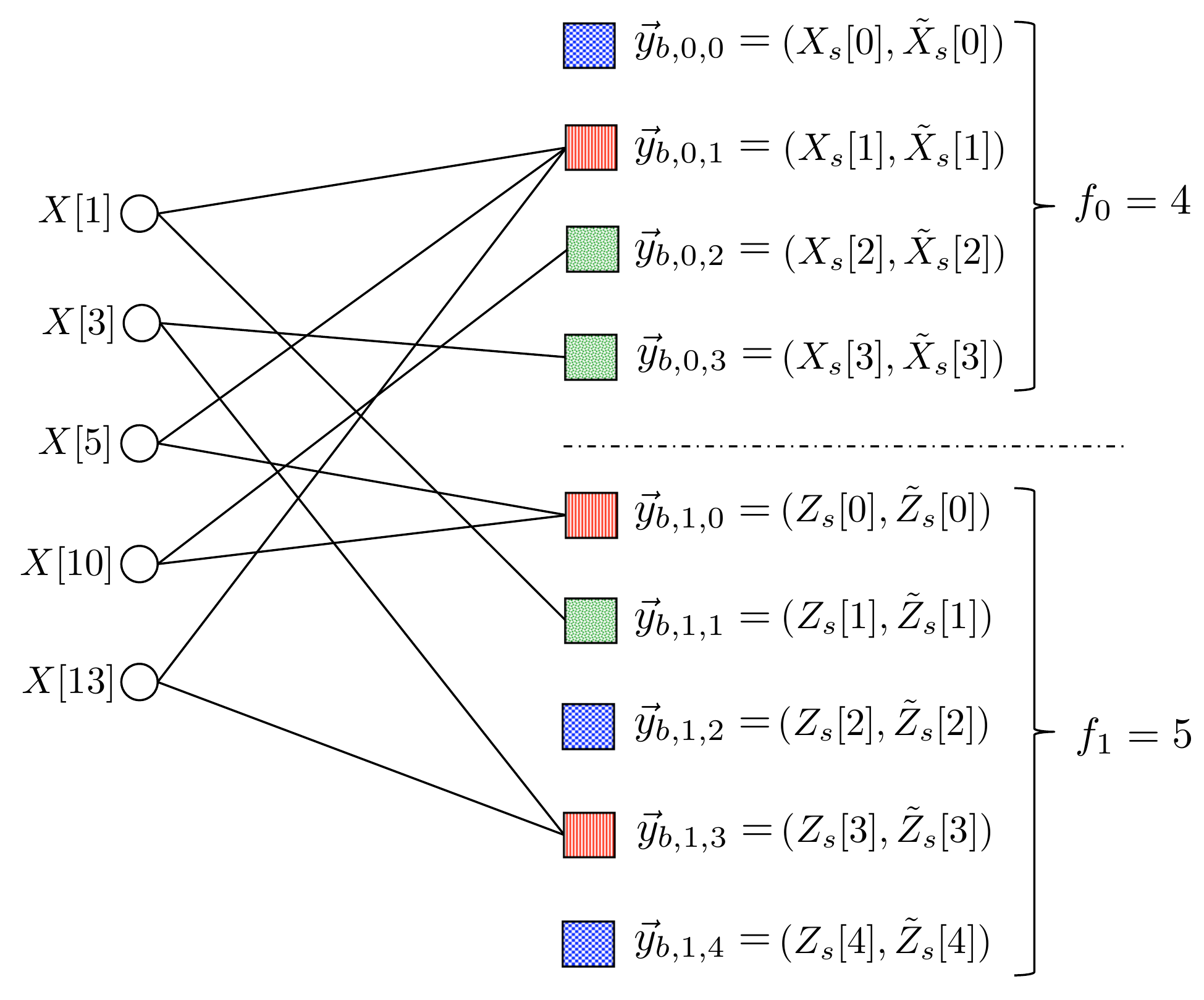}
\caption{A $2$-left regular degree bi-partite graph representing relation between the unknown non-zero DFT coefficients, of the $20$-length example signal $\signal$, and the bin observations obtained through the FFAST front-end architecture shown in Fig.~\ref{fig:ffastex}. Left nodes of the bi-partite graph represent the $5$ non-zero DFT coefficients of the input signal $\signal$, while the right nodes represent the ``bins" ( also sometimes referred in the sequel as ``check nodes") with vector observations. An edge connects a left node to a right node iff the corresponding non-zero DFT  coefficient contributes to the observation vector of that particular bin. The observation at each check node is a 2-dimensional complex-valued vector, e.g., $\binobsv{0}{0} = (X_s[0], \tilde{X}_s[0])$.}
\label{fig:aliasinggraph}
\end{center}
\end{figure}

In this section, we represent the relation between the original $20$-length $5$-sparse DFT $\transform$ and the FFAST front-end output, i.e., bin-observations, obtained using the Fourier properties, in a graphical form and interpret it as a ``channel-code". In particular, since this code is a result of a sub-sampling and aliasing operations, we refer to it as an ``Alias-code". We also establish that computing the sparse DFT $\transform$ of the signal $\signal$ from its sub-sampled data is equivalent to decoding the alias-code resulting from processing the input signal $\signal$ through the FFAST front-end.

Suppose that the $20$-length example input signal $\signal$ is processed through a $2$ stage FFAST front-end architecture shown in Fig.~\ref{fig:ffastex}, to obtain the bin-observation vectors $(X_s[\cdot],\tilde{X}_s[\cdot])$ and $(Z_s[\cdot],\tilde{Z}_s[\cdot])$. Then, the relation between the $9$ bin-observation vectors and the $5$ non-zero DFT coefficients of the signal $\signal$ can be computed using the signal processing properties. A graphical representation of this relation is shown in Fig.~\ref{fig:aliasinggraph}. Left nodes of the bi-partite graph represent the $5$ non-zero DFT coefficients of the input signal $\signal$, while the right nodes represent the ``bins" ( also sometimes referred in the sequel as ``check nodes") with vector observations. An edge connects a left node to a right node iff the corresponding non-zero DFT  coefficient contributes to the observation vector of that particular bin, e.g., after aliasing, due to sub-sampling, the DFT coefficient $\transformc{10}$ contributes to the observation vector of bin $2$ of stage $0$ and bin $0$ of stage $1$. Thus, the problem of computing a $5$-sparse $20$-length DFT has been transformed into that of decoding the values and the support of the left nodes of the bi-partite graph in Fig.~\ref{fig:aliasinggraph}, i.e., decoding of alias-code. Next, we classify the observation bins based on its edge degree in the bi-partite graph, i.e., number of non-zero DFT coefficients contributing to a bin, which is then used to decode the alias-code.
\begin{itemize}
\item {\bf zero-ton}: A bin that has no contribution from any of the non-zero DFT coefficients of the signal, e.g., bin $0$ of stage $0$ in Fig.~\ref{fig:aliasinggraph}. A zero-ton bin can be trivially identified from its observations.
\item {\bf single-ton}: A bin that has contribution from exactly one non-zero DFT coefficient of the signal, e.g., bin $2$ of stage $0$. Using the signal processing properties, it is easy to verify that the observation vector of bin $2$ of stage $0$ is given as,
\begin{equation*}
\binobsv{0}{2} = 
\left(
\begin{array}{c}
  \transformc{10}\\
  e^{2\pi\imath10/20}  \transformc{10}
  \end{array}
\right).
\end{equation*}
The observation vector of a single-ton bin can be used to determine the support and the value, of the only non-zero DFT coefficient contributing to that bin, as follows:
\begin{itemize}
\item {\em support}: The support of the non-zero DFT coefficient contributing to a single-ton bin can be computed as,
\begin{equation}
10 = \frac{20}{2\pi}\angle \binobsv{0}{2}[1]y^{\dagger}_{b,0,2}[0]
\end{equation}
\item {\em Value}: The value of the non-zero DFT coefficient is given by the observation $y_{b,0,2}[0]$.
\end{itemize}
We refer to this procedure as a ``ratio-test", in the sequel. Thus, a simple ratio-test on the observations of a single-ton bin correctly identifies the support and the value of the only non-zero DFT coefficient connected to that bin. It is easy to verify that this property holds for all the single-ton bins.

\item {\bf multi-ton}: A bin that has a contribution from more than one non-zero DFT coefficients of the signal, e.g., bin $1$ of stage $0$. The observation vector of bin $1$ of stage $0$ is,
\begin{eqnarray*}
\binobsv{0}{1} &=& \transformc{1}\left(
\begin{array}{c}
  1 \\
  e^{\imath2\pi/20}
\end{array}
\right) + \transformc{5}\left(
\begin{array}{c}
  1 \\
  e^{\imath2\pi5/20}
\end{array}
\right) + \transformc{13}\left(
\begin{array}{c}
  1 \\
  e^{\imath2\pi13/20}
\end{array}
\right)\\
&=& \left(\begin{array}{c}
  10 \\
  -3.1634 - \imath3.3541
\end{array}
\right)
\end{eqnarray*}
Now, if we perform the ``ratio-test" on these observations, we get, the support to be $12.59$. Since, we know that the support has to be an integer value between $0$ to $19$, we conclude that the observations do not correspond to a single-ton bin. In Appendix~\ref{app:multi-tonratiotest}, we rigorously show that the ratio-test identifies a multi-ton bin almost surely.
\end{itemize}

Hence, using the ``ratio-test" on the bin-observations, the FFAST decoder can determine if a bin is a zero-ton, a single-ton or a multi-ton, almost surely. Also, when a bin is single-ton the ratio-test provides the support as well as the value of the non-zero DFT coefficient connected to that bin. In the following Section~\ref{sec:FFASTdecoder}, we provide the FFAST peeling-decoder that computes the support and the values of all the non-zero DFT coefficients of the signal $\signal$.

\subsection{FFAST back-end peeling-decoder}\label{sec:FFASTdecoder}
In the previous section we have explained how the {\em ratio-test} can be used to determine if a bin is a zero-ton, single-ton or a multi-ton bin almost surely. Also, in the case of a single-ton bin the ratio-test also identifies the support and the value of the non-zero DFT coefficient connected to that bin. 

\paragraph*{FFAST peeling-decoder} The FFAST decoder repeats the following steps (the pseudocode is provided in Algorithm~\ref{alg:FFAST} and Algorithm~\ref{alg:single-ton_estimator} in Appendix~\ref{app:pseudocode}):
\begin{enumerate}
\item Select all the edges in the graph with right degree 1 (edges connected to single-ton bins).
\item Remove these edges from the graph as well as the associated left and right nodes.
\item Remove all the other edges that were connected to the left nodes removed in step-2. When a neighboring edge of any right node is removed, its contribution is subtracted from that check node.
\end{enumerate}
Decoding is successful if, at the end, all the edges have been removed from the graph. It is easy to verify that the FFAST peeling-decoder operated on the example graph of Fig.~\ref{fig:aliasinggraph} results in successful decoding, with the coefficients being uncovered in the following possible order: $\transformc{10},\transformc{3},\transformc{1},\transformc{5},\transformc{13}$.

Thus, the FFAST architecture has {\em transformed the problem of computing the DFT of $\signal$ into that of decoding alias-code} of Fig.~\ref{fig:aliasinggraph}. Clearly the success the FFAST decoder depends on the properties of the sparse bi-partite graph resulting from the sub-sampling operation of the FFAST front-end. In particular, if the sub-sampling induced aliasing bi-partite graph is {\em peeling-friendly}, i.e., has few single-ton bins to initiate the peeling procedure and creates new single-tons at each iteration, until all the DFT coefficients are uncovered, the FFAST peeling-decoder succeeds in computing the DFT $\transform$.

\subsection{Connection to coding for packet erasure channels}\label{sec:erasurechannel}
\begin{figure}[t]
 \begin{center}
\subfigure[Bi-partite code graph, corresponding to a parity check matrix of an $(\length,\length-\nbins)$ sparse-graph code designed for an erasure channel.]{
{\includegraphics[scale=0.45]{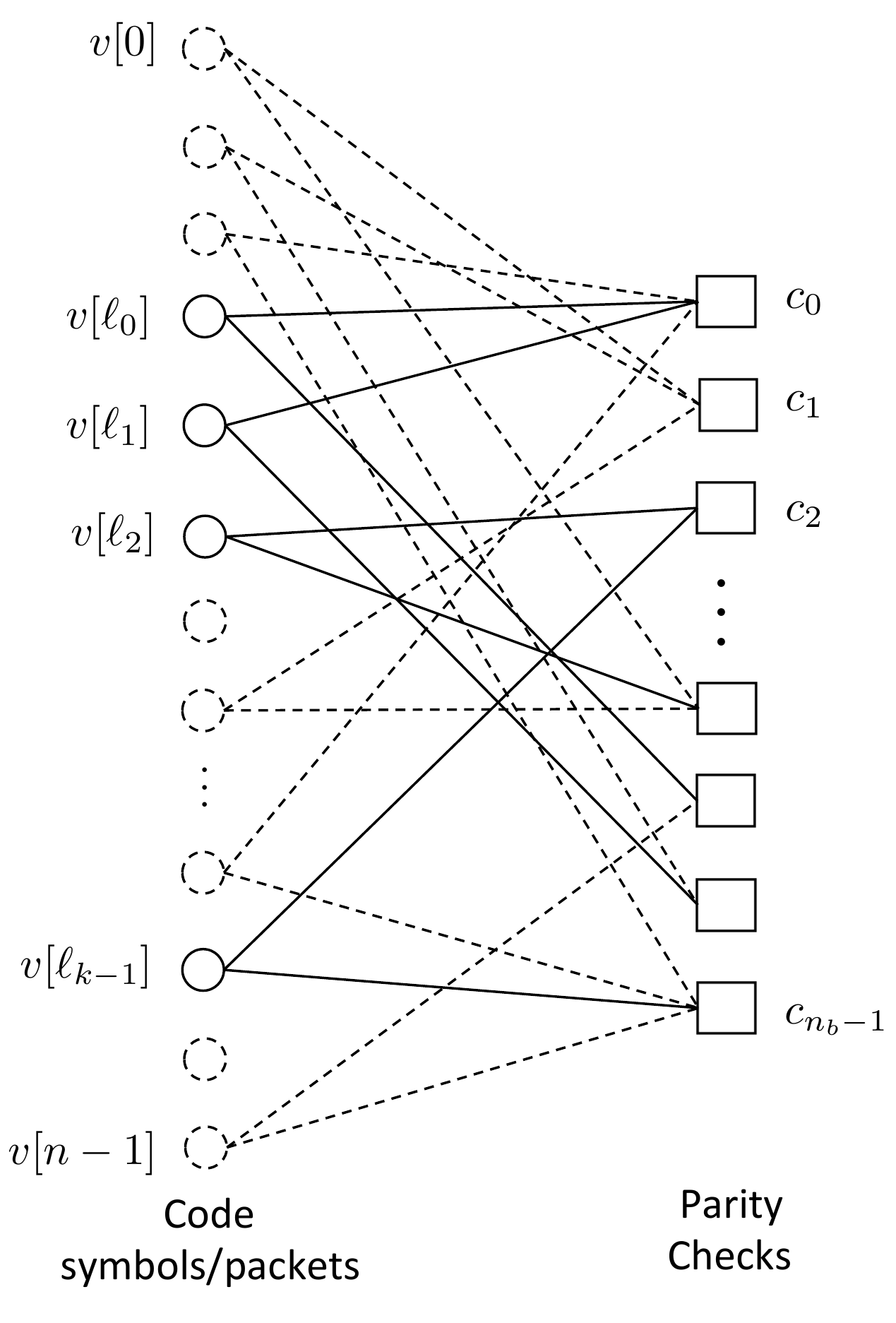}}
\label{fig:codegraph}
}\quad\quad\qquad
\subfigure[Bi-partite graph representing the aliasing connections, resulting from sub-sampling operation of the signal $\signal$ by an FFAST front-end.]{
\includegraphics[scale = 0.45]{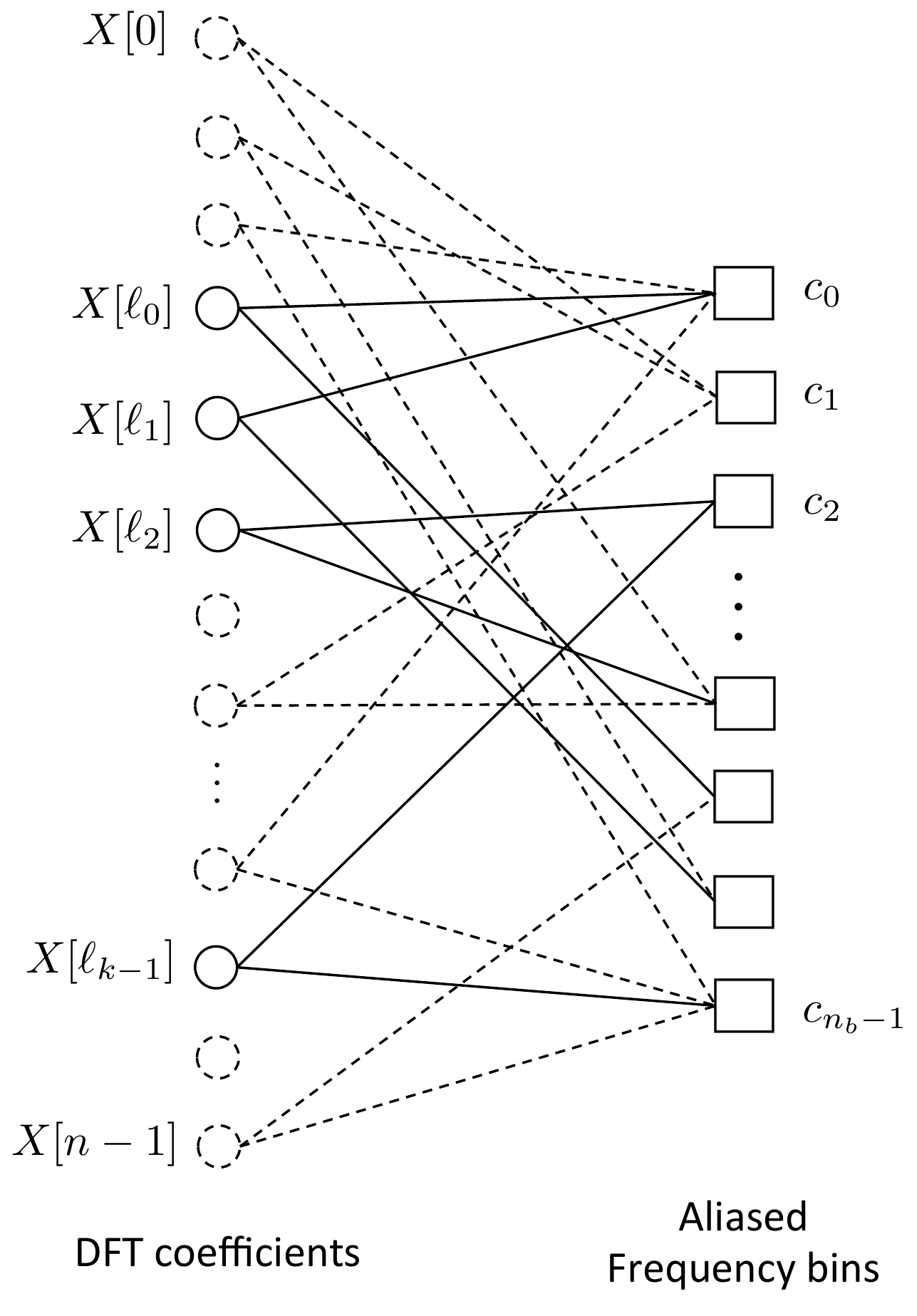}
\label{fig:FFASTgraph}
}
\end{center}
\caption{Comparison between the bi-partite graphs corresponding to the parity check matrix of a sparse-graph code for an erasure channel and a graph induced by the FFAST front-end subsampling architecture.}\label{fig:erasure connection}
\end{figure}

\begin{table}[h!]
\begin{center}
\begin{tabulary}{15cm}{|L|L|}
\hline
Erasure Channel & Sparse DFT\\
\hline
1. Explicitly designed sparse-graph code. & 1. Implicitly designed sparse-graph code induced by sub-sampling. \\
\hline 
2. $\length-\sparsity$ correctly received packets. & 2. $\length-\sparsity$ zero DFT coefficients.\\
\hline
3. $\sparsity$-erased packets. & 3. $\sparsity$ unknown non-zero DFT coefficients\\
\hline 
4. Peeling-decoder recovers the values of the erased packets using `single-ton' check nodes. The identity (location) of the erased packets is {\em known}. & 4. Peeling-decoder recovers {\em both} the values and the locations of the non-zero DFT coefficients using `single-ton' check nodes. The locations of the non-zero DFT coefficients are {\em not known}. This results in a $2\times$ cost in the sample complexity. \\
\hline
5. Codes based on regular-degree bipartite graphs are near-capacity-achieving. More efficient, capacity-achieving irregular-degree bipartite graph codes can be designed. & 5. The FFAST architecture based on uniform subsampling can induce only left-regular degree bi-partite graphs. \\
\hline
\end{tabulary}
\end{center}
\caption{Comparison between decoding over a sparse-graph code for a packet erasure channel and computing a sparse DFT using the FFAST architecture.}\label{tab:similarities}
\end{table}

The problem of decoding sparse-graph codes for erasure channel has been well studied in the coding theory literature. In this section we draw an analogy between decoding over sparse-graph codes for a packet erasure channel and decoding over sparse bi-partite graphs induced by the FFAST architecture. 

\paragraph{Sparse-graph code for a packet-erasure channel}
Consider an $(\length,\length-\nbins)$ packet erasure code. Each $\length$-length codeword consists of $(\length-\nbins)$ information packets and $\nbins$ parity packets. The erasure code is defined by a bi-partite graph as shown in Fig.~\ref{fig:codegraph}. An $\length$-length sequence of packets that satisfies the constraints defined by the graph in Fig.~\ref{fig:codegraph}, i.e., sum of the packets connected to a parity check node equals zero, is a valid codeword. Suppose a codeword from the code, defined by the graph of Fig.~\ref{fig:codegraph}, is transmitted over an erasure channel that uniformly at random drops some $\sparsity$ number of packets. In Fig.~\ref{fig:codegraph}, we use dotted circles to represent the correctly received packets (a cyclic redundancy check can be used to verify the correctly received  packets). Let $\{v[\ell_0],\hdots,v[\ell_{k-1}]\}$ be the $\sparsity$ packets that are dropped/erased by the channel. The dotted edges in the bi-partite graph of Fig.~\ref{fig:codegraph} denote the operation of subtracting the contribution of all the correctly received packets from the corresponding check nodes. A peeling-decoder can now iteratively unravel the erased packets that are connected to a check node with exactly one erased packet. If the bipartite graph consisting only of the solid variable nodes and the solid edges is such that the peeling-decoder successfully unravels all the erased packets, decoding is successful.

\paragraph{Decoding over a bipartite graph induced by the FFAST}
Consider an $\length$-length signal $\signal$ that has a $\sparsity$-sparse DFT $\transform$. Let the signal $\signal$ be sub-sampled by some appropriately designed FFAST front-end. The induced spectral-domain aliasing due to FFAST front-end sub-sampling operation is graphically represented by a bipartite graph shown in Fig.~\ref{fig:FFASTgraph}. The variable (left) nodes correspond to the $\length$-length DFT $\transform$ and the check (right) nodes are the bins consisting of the aliased DFT coefficients. A variable node is connected to a check node, iff after aliasing that particular DFT coefficient contributes to the observation of the considered check node. Let $\{\transformc{\ell_0},\transformc{\ell_1},\hdots, \transformc{\ell_{\sparsity-1}}\}$ be the $\sparsity$ non-zero DFT coefficients. The zero DFT coefficients are represented by dotted circles. The dashed edges in Fig.~\ref{fig:FFASTgraph} denotes that the contribution to the bin-observation, due to this particular edge is zero. Using the ratio-test on the vector observation at each check-node one can determine if the check node is a ``single-ton", i.e., has exactly one solid edge. A peeling-decoder can now iteratively unravel the non-zero DFT coefficients connected to single-ton check nodes. If the bipartite graph consisting only of the solid variable nodes and the solid edges is such that peeling-decoder successfully unravels all the variable nodes, the algorithm  succeeds in computing the DFT $\transform$. In Table~\ref{tab:similarities} we provide a comparison between decoding over bi-partite graphs of Fig.~\ref{fig:codegraph} and Fig.~\ref{fig:FFASTgraph}.

Thus, the problem of decoding bi-partite graphs corresponding to sparse-graph codes designed for a packet-erasure channel is closely related to decoding the sparse bi-partite graphs induced by the FFAST architecture. We use this analogy: a) to design a sub-sampling front-end that induces a `good' left-regular degree sparse-graph codes; and b) to formally connect our proposed Chinese-Remainder-Theorem based aliasing framework to a random sparse-graph code constructed using a balls-and-bins model (explained in Section~\ref{sec:bbensemble}), and analyze the convergence behavior of our algorithm using well-studied density evolution techniques from coding theory.

Next, we address the question of how to carefully design the sub-sampling parameters of the FFAST front-end architecture so as to induce ``good-graphs" or ``alias-codes". In Section~\ref{sec:verysparse} and Section~\ref{sec:lesssparse} we provide constructions of the FFAST front-end architecture and analyze the performance of the FFAST peeling-decoder for the {\em very-sparse}, i.e., $0 < \sindex \leq 1/3$, and the {\em less-sparse}, i.e., $1/3 < \sindex < 1$, regimes of sparsity respectively.
\section{FFAST Construction and Performance analysis for the {\em very-sparse} ($ \sparsity=O(\length^{\sindex}), \  0 < \sindex \leq 1/3$) regime}\label{sec:verysparse}
In Section~\ref{sec:DFTSparseCodes}, using an example, we illustrated that the problem of computing a $\sparsity$-sparse $\length$-length DFT of a signal can be transformed into a problem of decoding over sparse bipartite graphs using the FFAST architecture. In this section, we provide a choice of parameters of the FFAST front-end architecture and analyze the probability of success of the FFAST peeling-decoder for the very-sparse regime of $0 < \sindex < 1/3$. As shown in Section~\ref{sec:erasurechannel}, the FFAST decoding process is closely related to the decoding procedure on  sparse-graph codes designed for erasure channels. From the coding theory literature, we know that there exist several sparse-graph code constructions that are low-complexity and capacity-achieving for the erasure channels. The catch for us is that we are not at liberty to use any arbitrary bi-partite graph, but {\em can choose only those graphs that correspond to the alias-codes, i.e., are induced via aliasing through our proposed FFAST subsampling front-end}. How do we go about choosing the right parameters and inducing the good graphs?

We describe two ensembles of bi-partite graphs. The first ensemble is based on a ``balls-and-bins" model, while the second ensemble is based on the CRT. The balls-and-bins model based ensemble of graphs is closer in spirit to the sparse-graph codes in the coding-theory literature. Hence, is amenable to a rigorous analysis using coding-theoretic tools like density-evolution \cite{luby2001improved}. The bi-partite graphs induced by the FFAST front-end sub-sampling operation belong to the CRT ensemble. Later, in Lemma~\ref{lem:equivalence}, we show that the two ensembles are equivalent. Hence, the analysis of the balls-and-bins construction carries over to the FFAST. We start by setting up some notations and common parameters.

Consider a set $\setfactors = \{\factor{0},\hdots, \factor{\stages-1}\}$ of pairwise co-prime integers. Let the signal length $\length = {\cal P}\prod_{i=0}^{\stages-1}\factor{i}$, for some positive integer ${\cal P} \geq 1$, and $\nbins \triangleq \sum_{i=0}^{\stages-1} \factor{i}$. The integers $\factor{i}$'s are chosen such that they are approximately equal and we use $\binsize$ to denote this value. More precisely, $\factor{i} = \binsize + O(1)$, for $i=0,\hdots,\stages-1$, where $\binsize$ is an asymptotically large number. The $O(1)$ perturbation term in each $\factor{i}$ is used to obtain a set of co-prime integers\footnote{An example construction of an approximately equal sized $3$ co-prime integers can be obtained as follows. Let $\binsize = 2^{r_0}3^{r_1}5^{r_2}$ for any integers $r_0,r_1,r_2$ greater than $1$. Then, $f_0 = \binsize + 2, \factor{1} = \binsize + 3$ and $\factor{2} = \binsize + 5$ are co-prime integers.} approximately equal to ${\binsize}$. We construct a FFAST sub-sampling front-end architecture with $\stages$ stages. Each stage further has $2$ delay-chains (see Fig.~\ref{fig:conceptual} for reference). The sub-sampling period used in both delay-chains of stage $i$ is $\length/\factor{i}$ and hence the number of output samples is $\factor{i}$. The total number of input samples used by the FFAST algorithm is $\samples = 2\stages\binsize + O(1)$ (see Fig.~\ref{fig:conceptual}). In this section, we use $\binsize = \binexp \sparsity$, for some constant $\binexp > 0$. This results in a sparsity index $0 < \sindex \leq 1/\stages$, depending on the value of the integer ${\cal P}$.

\subsection{Ensemble $\bbensemble$ of bi-partite graphs constructed using  a ``Balls-and-Bins'' model}\label{sec:bbensemble} Bi-partite graphs in the ensemble $\bbensemble$, have $\sparsity$ {\em variable nodes} on the left and $\nbins$ {\em check nodes} on the right. Further, each variable (left) node is connected to $\stages$ right nodes, i.e., left-regular degree bi-partite graphs. An example graph from an ensemble $\bbensemble$, for $\setfactors = \{4,5\}, \stages = 2, \sparsity = 5$ and $\nbins=9$ is shown in Fig.~\ref{fig:aliasinggraph}. More generally, the ensemble $\bbensemble$ of $\stages$-left regular edge degree bipartite graphs constructed using a ``balls-and-bins" model is defined as follows. Set $\nbins = \sum_{i=0}^{\stages-1}\factor{i}$, where $\setfactors = \{\factor{i}\}_{i=0}^{\stages-1}$. Partition the set of $\nbins$ check nodes into $\stages$ subsets with the $i^{th}$ subset having $\factor{i}$ check nodes. For each variable node, choose one neighboring check node in each of the $\stages$ subsets, uniformly at random. The corresponding $\stages$-left regular degree bipartite graph is then defined by connecting the variable nodes with their neighboring check nodes by an undirected edge.  

An {\it edge} $e$ in the graph is represented as a pair of nodes $e = \{v, c\}$, where $v$ and $c$ are the variable and check nodes incident on the edge $e$. By a {\it directed} edge $\vec{e}$ we mean an ordered pair $(v,c)$ or $(c,v)$ corresponding to the edge $e = \{v, c\}$. A {\it path} in the graph is a directed sequence of directed edges $\vec{e_1}, \hdots, \vec{e_{t}}$ such that, if $\vec{e_i} = (u_i,u'_{i})$, then the $u'_i = u_{i+1}$ for $i=1,\hdots,t-1$. The length of the path is the number of directed edges in it, and we say that the path connecting $u_1$ to $u_{t}$ starts from $u_1$ and ends at $u_{t}$.
\subsubsection{Directed Neighborhood} The {\it directed neighborhood of depth $\ell$} of $\vec{e} = (v, c)$, denoted by ${\cal N}^{\ell}_{\vec{e}}$, is defined as the induced subgraph containing all the edges and nodes on paths $\vec{e_1},\hdots,\vec{e_{\ell}}$ starting at node $v$ such that $\vec{e_1} \neq \vec{e}$. An example of a directed neighborhood of depth $\ell=2$ is given in Fig.~\ref{fig:localtree}. If the induced sub-graph corresponding to the directed neighborhood ${\cal N}^{\ell}_{\vec{e}}$ is a tree then we say that the depth-$\ell$ neighborhood of the edge $\vec{e}$ is {\it tree-like}.

\subsection{Ensemble $\crtensemble$ of bipartite graphs constructed using the Chinese-Remainder-Theorem (CRT)}\label{sec:crtensemble} The ensemble $\crtensemble$ of $\stages$-left regular degree bipartite graphs, with $\sparsity$ variable nodes and $\nbins$ check nodes, is defined as follows. Partition the set of $\nbins$ check nodes into $\stages$ subsets with the $i^{th}$ subset having $\factor{i}$ check nodes (see Fig.~\ref{fig:aliasinggraph} for an example). Consider a set ${\cal I}$ of $\sparsity$ integers, where each element of the set ${\cal I}$ is between $0$ and $\length-1$. Assign the $\sparsity$ integers from the set ${\cal I}$ to the $\sparsity$ variable nodes in an arbitrary order. Label the check nodes in the set $i$ from $0$ to $\factor{i}-1$ for all $i =0,\hdots,\stages-1$. A $\stages$-left regular degree bi-partite graph with $\sparsity$ variable nodes and $\nbins$ check nodes, is then obtained by connecting a variable node with an associated integer $v$ to a check node $(v)_{\factor{i}}$ in the set $i$, for $i=0,\hdots,\stages-1$. The ensemble $\crtensemble$ is a collection of all the $\stages$-left regular degree bipartite graphs induced by all possible sets ${\cal I}$. 

\begin{lemma}\label{lem:equivalence} The ensemble of bipartite graphs $\bbensemble$ is identical to the ensemble $\crtensemble$.
\begin{proof} It is trivial to see that $\crtensemble \subset \bbensemble$. Next we show the reverse. Consider a graph ${\cal G}_1 \in \bbensemble$. Suppose, a variable node $v \in {\cal G}_1$ is connected to the check nodes numbered $\{r_i\}^{\stages-1}_{i=0}$. Then, using the CRT, one can find {\cal P} number of integer's `$q$' between $0$ and $\length-1$ such that $(q)_{\factor{i}} = r_i \  \forall i = 0,\hdots,\stages-1$. Thus, for every graph ${\cal G}_1 \in \bbensemble$, there exists a set ${\cal I}$ of $\sparsity$ integers, that will result in an identical graph using the CRT based construction. Hence, $\bbensemble= \crtensemble$. 
\end{proof}
\end{lemma}

Note that the modulo rule used to generate a graph in the ensemble $\crtensemble$ is same as the one used in equation~\eqref{eq:aliasing_pattern} of Section~\ref{sec:DFTSparseCodes}. Thus, the FFAST architecture of Fig.~\ref{fig:conceptual}, {\em generates graphs from the CRT ensemble $\crtensemble$}, where the indices ${\cal I}$ of the $\sparsity$ variable nodes correspond to the locations (or support) of the non-zero DFT coefficients\footnote{A set ${\cal I}$ with repeated elements corresponds to a signal with fewer than $\sparsity$ non-zero DFT coefficients.} of the signal $\signal$. Also, under the assumption that the support of the non-zero DFT coefficients of the  signal $\signal$ is uniformly random, the resulting aliasing graph is uniformly random choice from the ensemble $\crtensemble$. 

Next, we analyze the performance of the FFAST peeling-decoder over a uniformly random choice of a graph from the ensemble $\bbensemble$, which along with the Lemma~\ref{lem:equivalence}, provides a lower bound on the success performance of the FFAST decoder over graphs from the ensemble $\crtensemble$. Although the construction and the results described in this section are applicable to any value of $\stages$, we are particularly interested in the case when $\stages = 3$. For $\stages = 3$ and ${\cal P} = O(1)$, we achieve the sub-linear sparsity index $\sindex = 1/3$, while other values of $0< \sindex < 1/3$, are achieved using larger values of ${\cal P}$. 

\subsection{Performance analysis of the FFAST peeling-decoder on graphs from the ensemble $\bbensemble$}\label{sec:analysis_bbensemble}
In this section, we analyze the probability of success of the FFAST peeling-decoder, over a randomly chosen graph from the ensemble $\bbensemble$, after a fixed number of peeling iterations $\ell$. Our analysis follows exactly the arguments in \cite{luby2001improved} and \cite{richardson2001capacity}. Thus, one may be tempted to take the results from \cite{luby2001improved} ``off-the-shelf". However, we choose here to provide a detailed analysis for two reasons. First, our graph  construction in the ensemble $\bbensemble$ is different from that used in \cite{luby2001improved}, which results in some fairly important differences in the analysis, such as the expansion properties of the graphs, thus warranting an independent analysis. Secondly, we want to make the paper more self-contained and complete. 

We now provide a brief outline of the proof elements highlighting the main technical components needed to show that the FFAST peeling-decoder decodes all the non-zero DFT coefficients with high probability.

\begin{itemize}
\item { \em Density evolution:} We analyze the performance of the message-passing algorithm, over a typical graph from the ensemble, for $\ell$ iterations. First, we assume that a local neighborhood of depth $2\ell$ of every edge in a typical graph in the ensemble is tree-like, i.e., cycle-free. Under this assumption, all the messages between variable nodes and the check nodes, in the first $\ell$ rounds of the algorithm, are independent. Using this independence assumption, we derive a recursive equation that represents the expected evolution of the number of single-tons uncovered at each round for this typical graph.

\item { \em Convergence to the cycle-free, case:} Using a Doob martingale as in \cite{luby2001improved}, we show that a random graph from the ensemble, chosen as per nature's choice of the non-zero DFT coefficients, behaves like a ``typical" graph, i.e., $2\ell$-depth neighborhood of most of the edges in the graph is cycle-free, with high probability. This proves that for a random graph in $\bbensemble$, the FFAST peeling-decoder decodes all but an arbitrarily small fraction of the variable nodes with high probability in a constant number of iterations, $\ell$.

\item { \em Completing the decoding using the graph expansion property:} We first show that if a graph is an ``expander" (as will be defined later in Section~\ref{sec:expansion}), and the FFAST peeling-decoder successfully decodes all but a small fraction of the non-zero DFT coefficients, then it decodes all the non-zero DFT coefficients successfully. Next, we show that a random graph from the ensemble $\bbensemble$ is an expander with high probability.
\end{itemize}

\subsubsection{{\bf Density evolution for local tree-like view}}\label{sec:cyclefree}
\begin{figure}[t]
\begin{center}
\includegraphics[scale=.4]{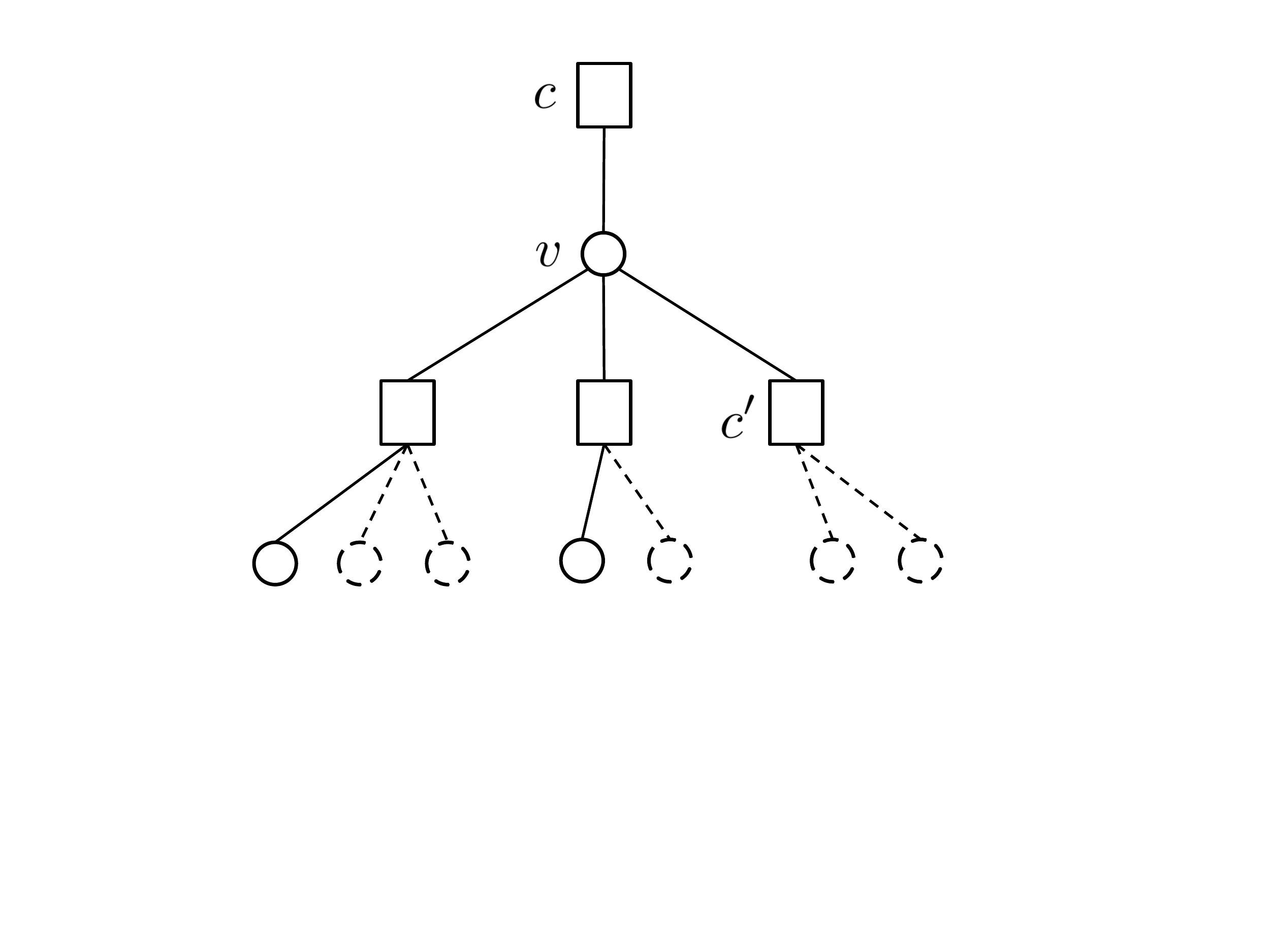}
\caption{Directed neighborhood of depth $2$ of an edge $\vec{e} = (v,c)$. The dashed lines correspond to nodes/edges removed at the end of iteration $j$. The edge between $v$ and $c$ can be potentially removed at iteration $j+1$ as one of the check nodes $c'$ is a single-ton (it has no more variable nodes remaining at the end of iteration $j$).}
\label{fig:localtree}
\end{center}
\end{figure}

In this section we assume that a local neighborhood of depth $2\ell$ of every edge in a graph in the ensemble is tree-like. Next, we define the edge-degree distribution polynomials of the bipartite graphs in the ensemble as $\lambda(\alpha) \triangleq \sum_{i=1}^{\infty} \lambda_i \alpha^{i-1}$ and $\rho(\alpha) \triangleq \sum_{i=1}^{\infty} \rho_i \alpha^{i-1}$, where $\lambda_i$ (resp. $\rho_i$) denotes the probability that an edge of the graph is connected to a left (resp. right) node of degree $i$. Thus for the ensemble $\bbensemble$, constructed using the balls-and-bins procedure, $\lambda(\alpha) = \alpha^{\stages-1}$ by construction. Further, as shown in Appendix~\ref{app:degdistribution}, the edge degree distribution $\rho(\alpha) = \exp(-(1-\alpha)/\binexp)$.

Let $p_j$ denote the probability that an edge is present (or undecoded) after round $j$ of the FFAST peeling-decoder, then $p_0 = 1$. Under the tree-like assumption, the probability $p_{j+1}$, is given as,
\begin{equation}\label{eq:evolution}
p_{j+1} = \lambda(1 - \rho(1-p_j)) \ \ j = 0,1,\hdots,\ell-1.
\end{equation}

Equation~\eqref{eq:evolution} can be understood as follows (also see Fig.~\ref{fig:localtree}): the tree-like assumption implies that, up to iteration $\ell$, messages on different edges are independent. Thus, the total probability, that at iteration $j+1$, a variable node $v$ is {\em decoded} due to a particular check node is given by $\rho(1-p_j) = \sum_{i=1}^{\infty} \rho_i (1-p_j)^{i-1}$ and similarly the total probability that none of the neighboring check nodes decode the variable node $v$ is $p_{j+1} = \lambda(1-\rho(1-p_i))$. Specializing equation~\eqref{eq:evolution} for the edge degree distributions of $\bbensemble$ we get,
\begin{align}\label{eq:evolution2}
   p_{j+1} = \left(1 - e^{-\frac{p_j}{\binexp}}\right)^{\stages-1}, \ \forall \ j = 0,1,\hdots,\ell-1
\end{align}
where $p_0 =1$. The evolution process of \eqref{eq:evolution2} asymptotically (in the number of iterations $\ell$) converges to $0$, for an appropriate choice of the parameter $\binexp$, e.g., see Table~\ref{tab:convergence}.

\begin{table}[h]
\begin{center}
\begin{tabular}{|c|c|c|c|c|c|c|c|c|c|c|c|}
  \hline
  \stages & 2 & 3& 4 & 5 & 6 &7& 8 & 9\\
  \hline
  $\binexp$ &1.0000 &   0.4073   & 0.3237 &   0.2850  &  0.2616   & 0.2456  &  0.2336  &  0.2244\\
  \hline
  $\stages\binexp$ & 2.0000 &   1.2219 &   1.2948 &   1.4250  &  1.5696  &  1.7192  &  1.8688   & 2.0196\\
  \hline
\end{tabular}
\end{center}
\caption{Minimum value of $\binexp$, required for the density evolution of \eqref{eq:evolution2} to converge asymptotically. The threshold value of $\binexp$ depends on the number of stages $\stages$.}\label{tab:convergence}
\end{table}

\subsubsection{{\bf Convergence to cycle-free case}}\label{sec:convergencetocyclefreecase} In the following Lemma~\ref{lem:concentration} we show; a) the expected behavior over all the graphs in the ensemble $\bbensemble$ converges to that of a cycle-free case, and b) with exponentially high probability, the proportion of the edges that are not decoded after $\ell$ iterations of the FFAST peeling-decoder is tightly concentrated around $p_{\ell}$, as defined in \eqref{eq:evolution2}.

\begin{lemma}[Convergence to Cycle-free case]\label{lem:concentration}
Over the probability space of all graphs $\bbensemble$, let $Z$ be the total number of edges that are not decoded after $\ell$ (an arbitrarily large but fixed) iterations of the FFAST peeling-decoder over a randomly chosen graph. Further, let $p_{\ell}$ be as given in the recursion~\eqref{eq:evolution2}. Then there exist constants $\beta$ and $\gamma$ such that for any $\epsilon_1 > 0$ and sufficiently large $\sparsity$ we have
\begin{eqnarray}\label{eq:convergence}
(a) &&\expectation[Z] < 2kd p_{\ell}.\\
(b) &&\ Pr\left(| Z - \expectation[Z]| > \sparsity \stages\epsilon_1 \right) < e^{-\beta\epsilon_1^2\sparsity^{1/(4\ell+1)}}.
\end{eqnarray}
\begin{proof}
Please see Appendix~\ref{app:proofconcentration}.
\end{proof}
\end{lemma}

\subsubsection{{\bf Successful Decoding using Expansion}}\label{sec:expansion} In the previous section we have shown that with high probability, the FFAST peeling-decoder decodes all but an arbitrarily small fraction of variable nodes. In this section, we show how to complete the decoding if the graph is a ``good-expander". Our problem requires the following definition of an ``expander-graph", which is somewhat different from conventional notions of an expander-graph in literature, e.g., {\em edge expander, vertex expander or spectral expander} graphs.

\begin{definition} [Expander graph]\label{def:expander} A $\stages$-left regular degree bipartite graph from ensemble $\bbensemble$, is called an $(\alpha,\beta,\stages)$ expander, if for all subsets $S$, of variable nodes, of size at most $\alpha \sparsity$, there exists a right neighborhood of $S$, i.e., $N_i(S)$, that satisfies $|N_i(S)| > \beta |S|$, for some $i=0,\hdots,\stages-1$.
\end{definition}

In the following lemma, we show that if a graph is an expander, and if the FFAST peeling-decoder successfully decodes all but a small fraction of the non-zero DFT coefficients, then it decodes all the non-zero DFT coefficients successfully. 
\begin{lemma}\label{lem:expsucceeds} Consider a graph from the ensemble $\bbensemble$, with $|{\setfactors}|=\stages$, that is an $(\alpha,1/2,\stages)$ expander for some $\alpha >0$. If the FFAST peeling-decoder over this graph succeeds in decoding all but at most $\alpha \sparsity$ variable nodes, then it decodes all the variable nodes.
\end{lemma}
\begin{proof} 
See Appendix~\ref{app:expsucceeds} 
\end{proof}

In Lemma~\ref{lem:ballsandbinsexp}, we show that most of the graphs in the ensemble $\bbensemble$ are expanders.
\begin{lemma}\label{lem:ballsandbinsexp} Consider a random graph from the ensemble $\bbensemble$, where $\stages \geq 3$. Then, all the subsets $S$ of the variable nodes, of the graph, satisfy $\max\{|N_i(S)|\}_{i=0}^{\stages-1} > |S|/2$,
\begin{enumerate}
\item[a)]  with probability at least $1 - e^{-\epsilon \sparsity\log(\nbins/\sparsity)}$, for sets of size $|S| = \alpha \sparsity$, for small enough $\alpha > 0$ and some $\epsilon > 0$.
\item[b)]  with probability at least $1 - O(1/\nbins)$, for sets of size $|S| = o(\sparsity)$.
\end{enumerate}
\end{lemma}
\begin{proof} 
See Appendix~\ref{app:expander_ballsandbins}
\end{proof}
The condition $\stages\geq 3$ is a necessary condition for part $(b)$ of Lemma~\ref{lem:ballsandbinsexp}. This can be seen as follows. Consider a random graph from the ensemble $\bbensemble$, where $|{\setfactors}| = \stages$. If any two variable nodes in the graph have the same set of $\stages$ neighboring check nodes, then the expander condition, for the set $S$ consisting of these two variable nodes, will not be satisfied. In a bi-partite graph from the ensemble $\bbensemble$, there are a total of $O(\sparsity^\stages)$ distinct sets of $\stages$ check nodes. Each of the $\sparsity$ variable nodes chooses a set of $\stages$ check nodes, uniformly at random and with replacement, from the total of $O(\sparsity^\stages)$ sets. If we draw $\sparsity$ integers uniformly at random between $0$ to $\length-1 $, the probability $Pr(\sparsity;\length)$ that at least two numbers are the same is given by,
\begin{equation}
Pr(\sparsity;\length) \approx 1 - e^{-\sparsity^2/2n}.
\end{equation}
This is also known as the {\em birthday paradox} or the {\em birthday problem} in literature \cite{mitzenmacher2005probability}. For a graph from the ensemble $\bbensemble$, we have $\length=O(\sparsity^\stages)$. Hence, if the number of stages $\stages\leq 2$, there is a constant probability that there exists a pair of variable nodes that share the same neighboring check nodes, in both the stages, thus violating the expander condition.

\begin{theorem}\label{thm:ballsandbinsmain} The FFAST peeling-decoder over a random graph from the ensemble $\bbensemble$, where $\stages\geq 3$ and $\binsize = \binexp \sparsity$:
\begin{itemize}
\item[a)] successfully uncovers all the variable nodes with probability at least $1 - O(1/\nbins)$;
\item[b)] successfully uncovers all but a vanishingly small fraction, i.e., $o(\sparsity)$, of the variable nodes with probability at least $1 - e^{-\beta\epsilon_1^2\sparsity^{1/(4\ell+1)}}$ for some constants $\beta, \epsilon_1 > 0$, and $\ell > 0$.
\end{itemize}
for an appropriate choice of the constant $\binexp$ as per Table~\ref{tab:convergence}.
\begin{proof} Consider a random graph from the ensemble $\bbensemble$. Let $Z$ be the number of the edges not decoded by the FFAST peeling-decoder in $\ell$ (large but fixed constant) iterations after processing this graph. Then, from recursion \eqref{eq:evolution2} and Lemma~\ref{lem:concentration}, for an appropriate choice of the constant $\binexp$ (as per Table~\ref{tab:convergence}), $Z \leq \alpha \sparsity$, for an arbitrarily small constant $\alpha > 0$, with probability at least $1 - e^{-\beta\epsilon_1^2\sparsity^{1/(4\ell+1)}}$. The result then follows from Lemmas \ref{lem:ballsandbinsexp} and \ref{lem:expsucceeds}.
\end{proof}
\end{theorem}

\begin{corollary}\label{cor:crtverysparse} The FFAST peeling-decoder over a random graph from the ensemble $\crtensemble$, where $\stages\geq 3$ and $\binsize = \binexp \sparsity$:
\begin{itemize}
\item[a)] successfully uncovers all the variable nodes with probability at least $1 - O(1/\nbins)$;
\item[b)] successfully uncovers all but a vanishingly small fraction, i.e., $o(\sparsity)$, of the variable nodes with probability at least $1 - e^{-\beta\epsilon_1^2\sparsity^{1/(4\ell+1)}}$ for some constants $\beta, \epsilon_1 > 0$, and $\ell > 0$.
\end{itemize}
for an appropriate choice of the constant $\binexp$ as per Table~\ref{tab:convergence}.
\begin{proof} Follows from equivalence of ensembles Lemma~\ref{lem:equivalence} and Theorem~\ref{thm:ballsandbinsmain}.
\end{proof}
\end{corollary}

\subsection{The FFAST front-end architecture parameters for achieving the sparsity index $0 < \sindex \leq 1/3$}\label{sec:verysparseconstruction} 
Consider a set $\setfactors = \{\factor{0},\factor{1},\factor{2}\}$ of a pairwise co-prime integers. The integers $\factor{i}$'s are such that they are approximately equal, i.e, $\factor{i} = \binsize + O(1)$, for $i=0,1,2$, where $\binsize = 0.4073\sparsity$ (see Table~\ref{tab:convergence}) is an asymptotically large number. Set the signal length $\length = {\cal P}\prod_{i=0}^{2}\factor{i}$, where, ${\cal P} = \binsize^a$, thus achieving the sparsity index $\sindex = 1/(3+a)$, for any positive constant $a > 0$. We construct a FFAST sub-sampling front-end with $\stages = 3$ stages, where each stage further has $2$ delay-chains (see Fig.~\ref{fig:conceptual}). The sub-sampling period used in the both the delay-chains of the $i^{th}$ stage is $\length/\factor{i}$. As a result, the number of samples at the output of each delay-chain of the $i^{th}$ stage is $\factor{i}$ for $i = 0,1,2$, i.e., the total number of samples $\samples$ used by the FFAST algorithm is $\samples = 2(\factor{0}+\factor{1}+\factor{2}) < 2.45\sparsity$.
 
\section{FFAST Construction and Performance analysis for the {\em less-sparse} ($ \sparsity=O(\length^{\sindex}), \  1/3 < \sindex < 1$) regime}\label{sec:lesssparse}
In the FFAST front-end architecture for the less-sparse regime, the integers in the set ${\setfactors} = \{\factor{0},\hdots,\factor{\stages-1}\}$, unlike for the very-sparse regime of Section~\ref{sec:verysparse}, are not pairwise co-prime. Instead, for the less-sparse regime ($ \sparsity =O(\length^{\sindex}), \  1/3 < \sindex < 1$) the relation between the integers $\factor{i}$'s is bit more involved. Hence, for ease of exposition, we adhere to the following approach:

\begin{itemize}
\item First, we describe the FFAST front-end construction and analyze performance of the FFAST decoding algorithm for a simple case of less-sparse regime of $\sindex = 2/3$. 
\item Then, in Section~\ref{sec:foranyd}, we provide a brief sketch of how to generalize the FFAST architecture of $\sindex = 2/3$, to $\sindex = 1-1/\stages$, for integer values of $\stages \geq 3$. This covers the range of values of $\sindex = 2/3, 3/4, \hdots$ etc. 
\item In Section~\ref{sec:intermediatedelta}, we show how to achieve the intermediate values of $1/3 < \sindex < 1$.
\item Finally, in Section~\ref{sec:lesssparseconstruction}, we use all the techniques learned in the previous sections to provide an explicit choice of parameters for the FFAST front-end architecture that achieves all the sparsity indices in the range $1/3 < \sindex < 1$. 
\end{itemize}
\begin{figure}[t]
\begin{center}
\includegraphics[scale=.5]{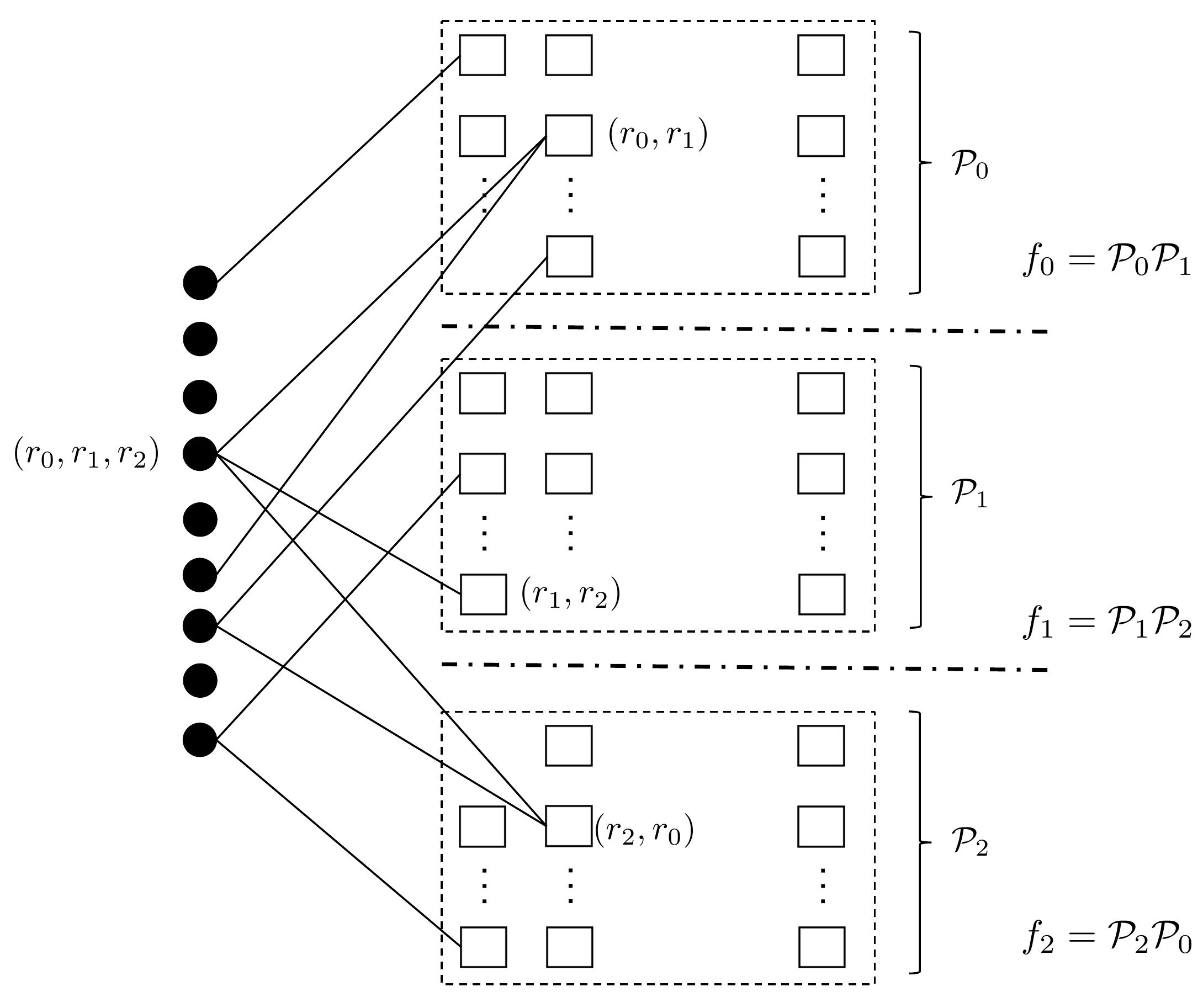}
\caption{A bi-partite graph with $\sparsity$ variable nodes and $\nbins = \sum_{i=0}^{2} \factor{i}$ check nodes, constructed using a balls-and-bins model. The check nodes in each of the $3$ sets are arranged in a matrix format, e.g., the $\factor{0}$ check nodes in the set $0$ are arranged in $\primef{0}$ rows and $\primef{1}$ columns. A check node in row $r_0$ and column $r_1$ in the set $0$, is indexed by a pair $(r_0,r_1)$ and so on and so forth for all the other check nodes. Each variable node chooses a triplet $(r_0,r_1,r_2)$, where $r_i$ is between $0$ and $\primef{i}-1$ uniformly at random. A $3$-regular degree bi-partite graph is then constructed by connecting a variable node with a triplet $(r_0, r_1, r_2)$ to check nodes $(r_0,r_1), (r_1,r_2)$ and $(r_2,r_0)$ in the three sets of check nodes respectively.}
\label{fig:cyclicshift}
\end{center}
\vspace*{-0.4in}
\end{figure}
\subsection{Less-sparse regime of $\sindex = 2/3$}\label{sec:twothirdrate}
\subsubsection{\bf FFAST front-end construction}\label{sec:FFASTlesssparse} Consider $\length = \prod_{i=0}^{2}\primef{i}$, where the set $\{\primef{i}\}_{i=0}^2$ consists of approximately equal sized co-prime integers with each $\primef{i} = \binsize + O(1)$ and $\binsize = \sqrt{\binexp\sparsity}$. This results in $\sindex =2/3$. Choose the integers $\{\factor{i}\}_{i=0}^2$, such that $\{\factor{0},\factor{1},\factor{2}\} = \{\primef{0}\primef{1}, \primef{1}\primef{2}, \primef{2}\primef{0}\}$. Then, we construct a $\stages = 3$ stage FFAST front-end architecture, where stage $i$ has two delay-chains each with a sub-sampling period of $\length/\factor{i}$ and $\factor{i}$ number of output samples.  

\subsubsection{\bf Performance analysis of the FFAST decoding algorithm} In order to analyze the performance of the FFAST decoding algorithm, we follow a similar approach as in Section~\ref{sec:verysparse} for the very-sparse regime of $0 < \sindex \leq 1/3$. We first provide an ensemble of bi-partite graphs constructed using a balls-and-bins model. Then, we provide CRT based ensemble of bi-partite graphs, that are generated by the FFAST front-end of Section~\ref{sec:FFASTlesssparse}. We show by construction these two ensembles are equivalent and analyze the performance of the FFAST peeling-decoder on a uniformly random graph from the balls-and-bins ensemble. 
\paragraph{{ Balls-and-Bins construction}}\label{sec:bb2}
We construct a bi-partite graph with $\sparsity$ variable nodes on the left and $\nbins = \sum_{i=0}^{2} \factor{i}$, check nodes on the right (see Fig.~\ref{fig:cyclicshift}) using balls-and-bins model as follows. Partition the $\nbins$ check nodes into $3$ sets/stages containing $\factor{0},\factor{1}$ and $\factor{2}$ check nodes respectively. The check nodes in each of the $3$ sets are arranged in a matrix format as shown in Fig.~\ref{fig:cyclicshift}, e.g., $\factor{0}$ check nodes in the set $0$ are arranged as $\primef{0}$ rows and $\primef{1}$ columns. A check node in row $r_0$ and column $r_1$ in the set $0$, is indexed by a pair $(r_0,r_1)$ and so on and so forth. Each variable node uniformly randomly chooses a triplet $(r_0,r_1,r_2)$, where $r_i$ is between $0$ and $\primef{i}-1$. The triplets are chosen with replacement and independently across all the $\sparsity$ variable nodes. A $3$-regular degree bi-partite graph with $\sparsity$ variable nodes and $\nbins$ check nodes is then constructed by connecting a variable node with a triplet $(r_0, r_1, r_2)$ to the check nodes $(r_0,r_1), (r_1,r_2)$ and $(r_2,r_0)$ in the three sets on right respectively, as shown in Fig.~\ref{fig:cyclicshift}.

\paragraph{{CRT based bi-partite graphs induced by the FFAST architecture}} Each variable node is associated with an integer $v$ between $0$ and $\length-1$ (location of the DFT coefficient). As a result of the subsampling and computing a smaller DFTs in the FFAST architecture (see Fig~\ref{fig:conceptual}), a variable node with an index $v$ is connected to the check nodes $(v)_{\factor{0}}, (v)_{\factor{1}}$ and $(v)_{\factor{2}}$ in the 3 stages, in the resulting aliased bi-partite graph. The CRT implies that $v$ is uniquely represented by a triplet $(r_0,r_1,r_2)$, where $r_i = (v)_{\primef{i}}$. Also, $((v)_{\factor{i}})_{\primef{i}} = (v)_{\primef{i}} = r_i$, for all $i=0,1,2$. Thus, the FFAST architecture induces a $3$-regular degree bi-partite graph with $\sparsity$ variable nodes and $\nbins$ check nodes, where a variable node with an associated triplet $(r_0, r_1, r_2)$ is connected to the check nodes $(r_0,r_1), (r_1,r_2)$ and $(r_2,r_0)$ in the three sets respectively. Further, a uniformly random model for the support $v$ of the non-zero DFT coefficients, corresponds to choosing the triplet $(r_0,r_1,r_2)$ uniformly at random. Thus, by construction the ensemble of bi-partite graphs generated by the FFAST front-end is equivalent to the balls-and-bins construction. 

Following the outline of the proof of Theorem~\ref{thm:main} (provided in Section~\ref{sec:verysparse}), we can show the following:
\begin{enumerate}
\item {\em Density evolution for the cycle-free case}: Assuming a local tree-like neighborhood derive a recursive equation (similar to equation~\ref{eq:evolution2}) representing the expected evolution of the number of single-tons uncovered at each round for a ``typical" graph from the ensemble. 
\item {\em Convergence to the cycle-free case}: Using a Doob martingale show an equivalent of Lemma~\ref{lem:concentration} for the less-sparse regime, where the number of check nodes in the $3$ different stages $\factor{0},\factor{1}$ and $\factor{2}$ are not pairwise co-prime. 
\item {\em Completing the decoding using the graph expansion property}: A random graph from the ensemble is a good expander with high probability. Hence, if the FFAST decoder successfully decodes all but a constant fraction of variable nodes, it decodes all the variable nodes. 
\end{enumerate}

The analysis of the first two items for the less-sparse regime is similar in spirit to the one in Section~\ref{sec:verysparse}, and will be skipped here. However, the analysis of the third item will be described here as there are some key differences, mainly arising due to fact that integers in the set $\{\factor{i}\}_{i=0}^2$ are not co-prime as in Section~\ref{sec:verysparse}. For the very-sparse regime we have shown (in Lemma~\ref{lem:ballsandbinsexp}) that the bottleneck failure event is; not being able to decode {\em all} the DFT coefficients. Hence, in this section, we analyze this bottleneck failure event for the case of the less-sparse regime. In particular, we show that if the FFAST decoder has successfully decoded all but a small constant number of DFT coefficients, then it decodes all the DFT coefficients successfully with high probability.

\paragraph{{ Decoding all the variable nodes using the expansion properties of the CRT construction}}\label{sec:cubetrappingset}
\begin{figure}[h]
\begin{center}
\includegraphics[scale=.5]{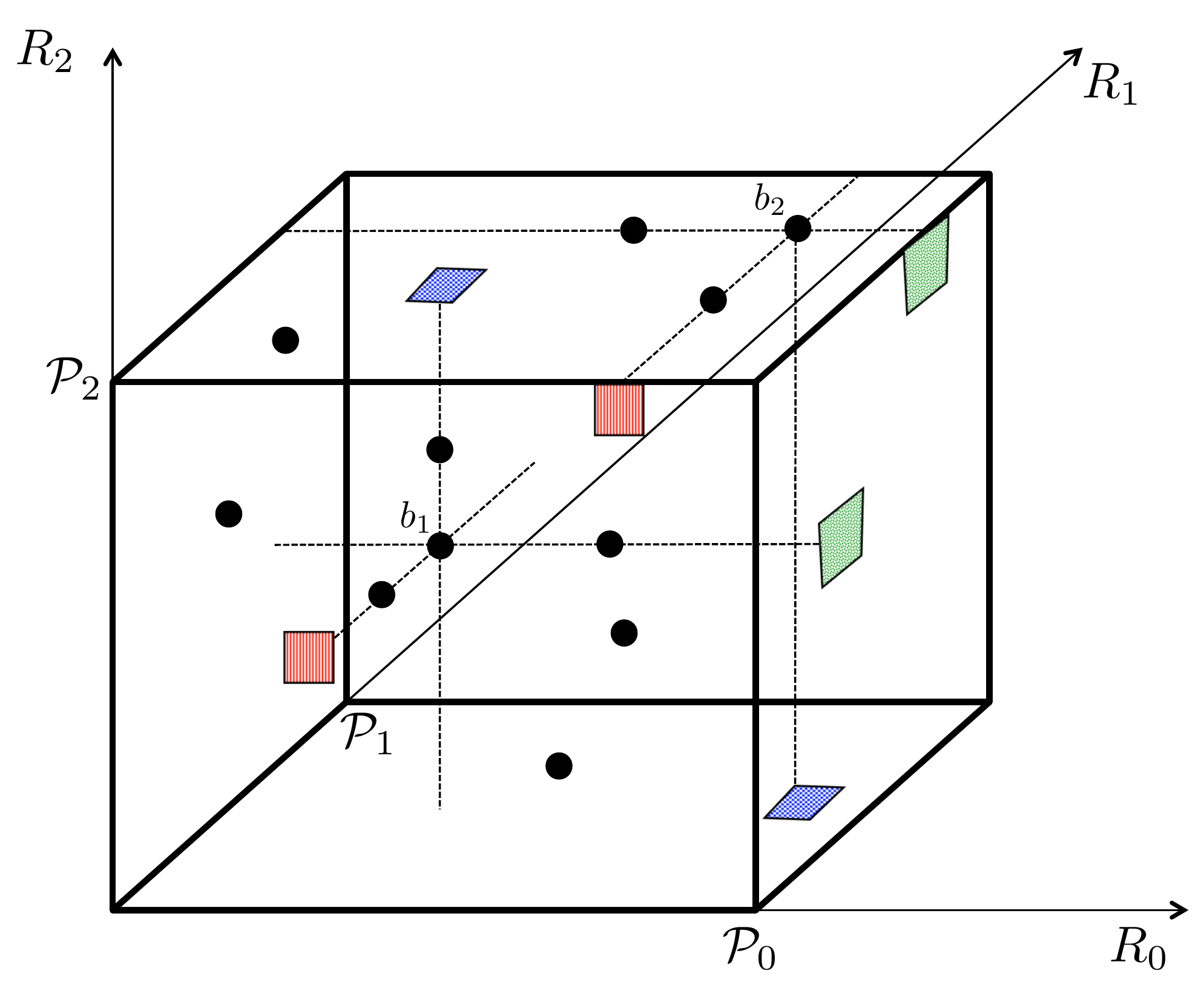}
\caption{A 3D visualization of the bi-partite graph in Fig.~\ref{fig:cyclicshift}, that belongs to the CRT ensemble corresponding to $\sindex = 2/3$. Recall that for $\sindex = 2/3$, $\length = \prod_{i=0}^{2}\primef{i}$, where the set $\{\primef{i}\}_{i=0}^2$ consists of approximately equal sized co-prime integers with each $\primef{i} = \binsize + O(1)$ and $\binsize = \sqrt{\binexp\sparsity}$. The integers $\{\factor{i}\}_{i=0}^2$, are such that $\{\factor{0},\factor{1},\factor{2}\} = \{\primef{0}\primef{1}, \primef{1}\primef{2}, \primef{2}\primef{0}\}$. A variable node with an associated triplet $(r_0,r_1,r_2)$ is represented by a `ball' at the position $(r_0,r_1,r_2)$. The $\factor{0}$ check nodes in the stage-1 of the bi-partite graph are represented by `blue' (checkered pattern) squares and likewise the ones in $\factor{1}$ are `green' (dotted pattern) and the check nodes in stage $\factor{2}$ are `red' (lines pattern). All the neighboring check nodes of a variable node, e.g., $b_1$, are multi-ton iff there is at least one more variable node along each of the three directions $R_0,R_1$ and $R_2$. The green (dotted pattern) and red (lines pattern) neighboring check nodes connected to the ball $b_2$ are multi-tons, while the blue (checkered pattern) neighboring check node is a single-ton since there are no other variable nodes along the $R_2$ direction of $b_2$.}
\label{fig:3Dcube}
\end{center}
\end{figure}

Consider an alternative visualization of the bi-partite graph in Fig.~\ref{fig:cyclicshift}, as shown in Fig.~\ref{fig:3Dcube}. A variable node associated with a triplet $(r_0,r_1,r_2)$ is represented by a ball at the position $(r_0,r_1,r_2)$. The plane $R_0$-$R_1$ corresponds to the check nodes in the stage $\factor{0}$, in a sense that all the variable nodes that have identical $(r_0,r_1)$ but distinct $r_2$ are connected to the check node $(r_0,r_1)$ and so on. Similarly the planes $R_1$-$R_2$ and $R_2$-$R_0$ correspond to the check nodes in stages $\factor{1}$ and $\factor{2}$ respectively. Thus, a variable node with co-ordinates $(r_0,r_1,r_2)$ is connected to multi-ton check nodes, if and only if there exist variable nodes with co-ordinates $(r_0,r_1,\osratio'_2)$, $(\osratio'_0,r_1,r_2)$ and $(r_0,\osratio'_1,r_2)$ (see Fig.~\ref{fig:3Dcube}), i.e., one neighbor in each axis. The FFAST decoder stops decoding if all the check nodes are multi-tons. Next, we find an upper bound on the probability of this `bad' event.

Consider a set $S$ of variable nodes such that $|S| = s$, where $s$ is a small constant. Let $E_S$ be an event that all the neighboring check nodes of all the variable nodes in the set $S$ are multi-tons, i.e., the FFAST decoder fails to decode the set $S$. Also, let $E$ be an event that there exists such a set. We first compute an upper bound on the probability of the event $E_S$, and then apply a union bound over all $\sparsity \choose s$ sets to get an upper bound on the probability of the event $E$.

Each variable node in the set $S$ chooses an integer triplet $(r_0,r_1,r_2)$ uniformly at random in a cube of size $\primef{0} \times \primef{1} \times \primef{2}$. Let $p_{\max}$ denote the maximum number of distinct values taken by these $s$ variable nodes on any co-ordinate. The FFAST decoder fails to decode the set $S$ if and only if all the variable nodes have at least one neighbor along each of the 3 directions $R_0,R_1,R_2$ (see Fig.~\ref{fig:3Dcube}). This implies that $s \geq 4p_{\max}$. Also, $p_{\max} > 1$, i.e., $s \geq 8$, since by the CRT all the variable nodes $s$ (with distinct associated integers) cannot have an identical triplet $(r_0,r_1,r_2)$. An upper bound on the probability of the event $E_S$ is then obtained as follows:
\begin{eqnarray}\label{eq:errorEvent1}
Pr(E_S) &<& \prod_{i=0}^2 \left(\frac{s}{4\primef{i}}\right)^s {\primef{i} \choose s/4}\nn\\
&\approx&\left(\frac{s}{4{\binsize}}\right)^{3s} {\binsize \choose s/4}^3\nonumber\\
&\overset{(a)}{<}&\left(\frac{s}{4{\binsize}}\right)^{3s} \left(\frac{4\binsize e}{s}\right)^{3s/4}\nonumber\\
&=& \left(\frac{se^{1/3}}{4{\binsize}}\right)^{9s/4},
\end{eqnarray}
where in $(a)$ we used ${p \choose q} \leq (p^q/q!) \leq (pe/q)^q$. Then, using a union bound over all possible ${\sparsity \choose s}$ sets, we get:
\begin{eqnarray}
Pr(E) &<& Pr(E_S){\sparsity \choose s} \nonumber\\
&<& \left(\frac{se^{1/3}}{4{\binsize}}\right)^{9s/4}\left(\frac{\sparsity e}{s}\right)^{s}\nonumber\\
&=& O(1/\nbins),
\end{eqnarray}
where in the last inequality, we used $s \geq 8$, $\sparsity =O(\binsize^2)$ and $\nbins = O(\binsize^2)$. 

Thus, the FFAST decoder decodes all the variable nodes with probability at least $1- O(1/\nbins)$.

\subsection{FFAST front-end architecture for $\sindex = 1-1/\stages$, where $\stages \geq 3$ is an integer}\label{sec:foranyd} 
Consider $\length = \prod_{i=0}^{\stages-1}\primef{i}$, where the set $\{\primef{i}\}_{i=0}^{\stages-1}$ consists of approximately equal sized co-prime integers with each $\primef{i} = \binsize + O(1)$, and $\binsize^{\stages-1} = \binexp\sparsity$, for a constant $\binexp$ chosen as per Table~\ref{tab:convergence}. This results in $\sindex =1 - 1/\stages$. Choose the integers $\{\factor{i}\}_{i=0}^{\stages-1}$, such that $\factor{i} = \primef{(i)_{\stages}}\primef{(i+1)_{\stages}}\cdots\primef{(i+{\stages-2})_{\stages}}$, for $i=0,\hdots,\stages-1$. Then, we construct a $\stages$ stage FFAST front-end architecture, where stage $i$ has two delay-chains each with a sub-sampling period of $\length/\factor{i}$ and $\factor{i}$ number of output samples. The performance of the back-end FFAST peeling-decoder, for these constructions, can be analyzed following the outline of Section~\ref{sec:verysparse} and Section~\ref{sec:lesssparse}. For $\sindex = 2/3$, the FFAST front-end architecture had $\stages = 3$ stages and as shown in Section~\ref{sec:twothirdrate}, the bottleneck failure event was to decode a set $S$ of size $|S| = 8 = 2^{3}$. For a general $\stages$ stage FFAST front-end construction, using induction, one can show that the worst case failure event is when the FFAST decoder fails to decode a set of size $|S| = 2^\stages$. The probability of this event is upper bounded by $1/{\binsize^{2^\stages - 2d}}$.

\subsection{Achieving the intermediate values of $\sindex$}\label{sec:intermediatedelta}
\begin{figure}[t]
\begin{center}
\includegraphics[scale=.5]{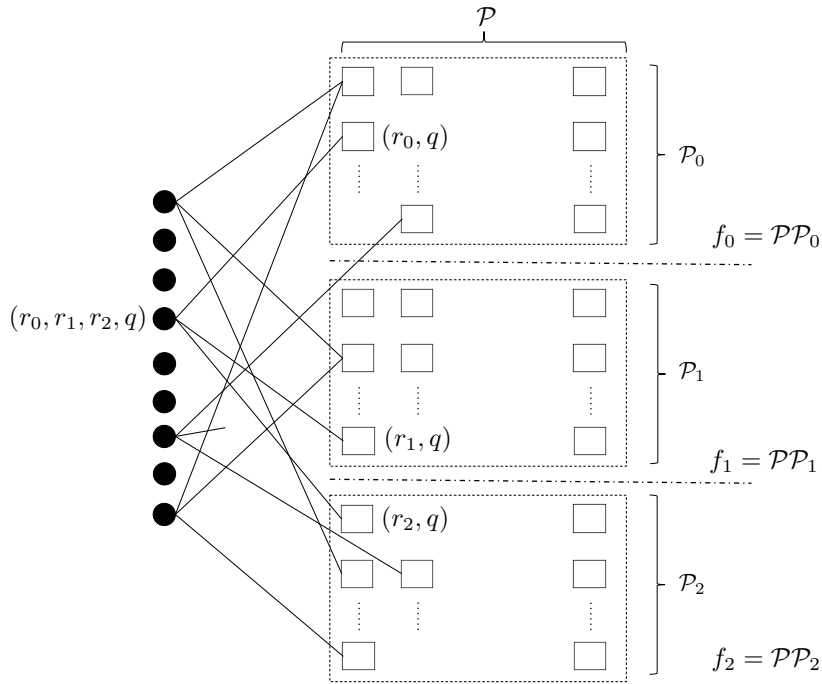}
\caption{Bi-partite aliasing graph resulting from processing an $\length = {\cal P}\prod_{i=0}^{2}\primef{i}$, length signal, using a $\stages = 3$ stage FFAST front-end. The set $\{\primef{i}\}_{i=0}^{2}$ consists of approximately equal sized co-prime integers with each $\primef{i} = \binsize + O(1)$ and ${\cal P} = \binsize^{a}$ for some constant $a > 0$. Further, the integers $\{\factor{0}, \factor{1}, \factor{2}\} = \{{\cal P}\primef{0}, {\cal P}\primef{1}, {\cal P}\primef{2}\}$. The check nodes in each of the $3$ sets are arranged so that a $j^{th}$ check node in the set $i$, belongs to the row $(j)_ {\primef{i}}$ and the column $\floor{j/{\cal P}}$. A variable node with an associated integer $v$ is uniquely represented by a quadruplet $(r_0,r_1,r_2,q)$, where $r_i = (v)_{\primef{i}}, \ i = 0,1,2$ and $q = \floor{v/{\cal P}}$, and is connected to check node $(r_i,q)$, i.e., check node in row $r_i$ and column $q$, of set $i$. The bipartite graph is a union of ${\cal P}$ disjoint bi-partite graphs.}
\label{fig:parallelview}
\end{center}
\vspace*{-0.3in}
\end{figure}
In this section, we show how to extend the scheme in Section~\ref{sec:verysparse}, that was designed for $\sparsity=O(\length^{1/3})$, to achieve a sparsity regime of $\sparsity =O(\length^{(1+a)/(3+a)})$ for some $a > 0$. This extension technique can essentially be used in conjunction with any of the operating points described earlier. Thus covering the full range of sparsity indices in $0 < \sindex < 1$.

Consider $\length = {\cal P}\prod_{i=0}^{2}\primef{i}$, where the set $\{\primef{i}\}_{i=0}^{2}$ consists of approximately equal sized co-prime integers with each $\primef{i} = \binsize + O(1)$ and ${\cal P} = \binsize^{a}$ for a constant $a > 0$. The parameter $\binsize^{1+a} = \binexp\sparsity$, for a constant $\binexp$ chosen as per Table~\ref{tab:convergence}. This results in $\sindex =(1+a)/(3+a)$. Further, choose the integers $\{\factor{0}, \factor{1}, \factor{2}\} = \{{\cal P}\primef{0}, {\cal P}\primef{1}, {\cal P}\primef{2}\}$. Then, we construct a $\stages=3$ stage FFAST front-end architecture, where stage $i$ has two delay-chains each with a sub-sampling period of $\length/\factor{i}$ and $\factor{i}$ number of output samples. 

\subsubsection{Union of Disjoint problems} The check nodes in each of the $3$ sets are arranged so that a $j^{th}$ check node in the set $i$, belongs to the row $(j)_ {\primef{i}}$ and the column $\floor{j/{\cal P}}$ (see Fig.~\ref{fig:parallelview}). 

A variable node with an associated integer $v$ is uniquely represented by a quadruplet $(r_0,r_1,r_2,q)$, where $r_i = (v)_{\primef{i}}, \ i = 0,1,2$ and $q = \floor{v/{\cal P}}$, and is connected to check node $(r_i,q)$, i.e., check node in row $r_i$ and column $q$, of set $i$. The resulting bipartite graph is a union of ${\cal P}$ disjoint bi-partite graphs, where each bi-partite subgraph behaves as an instance of the $3$-stage perfect co-prime case discussed in Section~\ref{sec:verysparse}. Then, using a union bound over these disjoint graphs one can compute the probability of the FFAST decoder successfully decoding all the variable nodes for asymptotic values of $\sparsity,\length$.

\subsection{FFAST architecture for achieving the sparsity index $1/3 < \sindex \leq 1$}\label{sec:lesssparseconstruction} 
In this section, we provide a sketch of the parameters used for the FFAST construction to achieve all values of sparsity index $\sindex$ in the less-sparse range, i.e., $1/3 < \sindex < 1$. The proposed construction is built based on the design principles illustrated in Sections~\ref{sec:twothirdrate},\ref{sec:foranyd} and Section~\ref{sec:intermediatedelta} and is an achievable construction. There are many other construction parameters of FFAST that may achieve similar sparsity index with equal or better performance. 
\subsubsection{$1/3 < \sindex < 0.73$} Consider $\length = {\cal P}\prod_{i=0}^{5}\primef{i}$, where the set $\{\primef{i}\}_{i=0}^5$ consists of approximately equal sized co-prime integers with each $\primef{i} = \binsize + O(1)$, and ${\cal P} = \binsize^{a}$, for $0 < a \leq 9$. Further, choose $\binsize^{2+ a} = \binexp\sparsity$, where constant $\binexp$ is given as per Table~\ref{tab:convergence}. This results in $1/3 < \sindex < 0.73$. We construct a $\stages = 6$ stage FFAST construction such that stage $i$ has two delay-chains each with a sub-sampling period of $\length/\factor{i}$ and $\factor{i}$ number of output samples. The integers $\{\factor{i}\}_{i=0}^5$ are such that $\{\factor{0},\factor{1},\factor{2}\} = \{{\cal P}\primef{0}\primef{1}, {\cal P}\primef{1}\primef{2}, {\cal P}\primef{2}\primef{0}\}$ and  $\{\factor{3},\factor{4},\factor{5}\} = \{{\cal P}\primef{3}\primef{4},{\cal P} \primef{4}\primef{5}, {\cal P}\primef{5}\primef{3}\}$. Note that for $a = 0$, this choice of parameters corresponds to two disjoint designs of $\sindex = 2/3$ the FFAST construction described in Section~\ref{sec:twothirdrate}. Using the analysis from Section~\ref{sec:twothirdrate}, along with a union bound over ${\cal P}$ sets of disjoint problems, as explained in Section~\ref{sec:intermediatedelta}, it is easy to verify that all $\sindex$ between $1/3 < \sindex < 0.73$, can be achieved with probability of error bounded above by $O(1/\nbins)$, if $\binexp \geq 0.2616$. Thus, the total number of samples $\samples$ used by the FFAST algorithm, for this range of sparsity index $\sindex$, is $\samples = 2\times6\times\binexp\sparsity \leq 3.14\sparsity$.
\subsubsection{$0.73 < \sindex < 0.875$} For the sparsity index $\sindex$ in the range $0.73 < \sindex < 0.875$, one possible construction is to use a $\stages = 8$ stage FFAST construction consisting of two disjoint designs for $\sindex = 4/3$, with $3$ cyclically shifted primes at a time. For example, consider a set $\{\primef{i}\}_{i=0}^7$ sof $8$ co-prime integers of approximately equal size then, $\{\factor{0},\factor{1},\factor{2},\factor{3}\} = \{{\cal P}\primef{0}\primef{1}\primef{2}, {\cal P}\primef{1}\primef{2}\primef{3}, {\cal P}\primef{2}\primef{3}\primef{0}, {\cal P}\primef{3}\primef{0}\primef{1}\}$ and  $\{\factor{4},\factor{5},\factor{6},\factor{7}\} = \{{\cal P}\primef{4}\primef{5}\primef{6}, {\cal P}\primef{5}\primef{6}\primef{7}, {\cal P}\primef{6}\primef{7}\primef{4}, {\cal P}\primef{7}\primef{4}\primef{5}\}$. For this construction, the parameter $a$ in ${\cal P} = \binsize^{a}$, is set between $10.5 < a < 32$, to achieve an error performance that asymptotically approaches zero as $ O(1/\nbins^{0.9})$. From Table~\ref{tab:convergence}, for a $\stages=8$ stage FFAST architecture $\binexp = 0.2336$. Thus, the total sample complexity for this design is $\samples = 2\stages\binexp\sparsity < 3.74\sparsity$.
\subsubsection{$0.875 < \sindex < 1$} The sparsity index $\sindex$ in the range $0.875 < \sindex < 1$, can be achieved by the combination of designs proposed in Section~\ref{sec:foranyd} and Section~\ref{sec:intermediatedelta} for increasing (but constant) values of $\stages$ as dictated by $\sindex$. For example, for $0.875 < \sindex < 0.99$, a $\stages = 8$ stage FFAST front-end design as in Section~\ref{sec:foranyd} along with approach of Section~\ref{sec:intermediatedelta} achieves an error performance that asymptotically approaches zero with $\samples < 3.74\sparsity$ as shown in Fig.~\ref{fig:constants}.

\section{Sample and computational complexity of the FFAST algorithm}\label{sec:complexity}
The FFAST algorithm performs the following operations in order to compute the $\length$-length DFT of an $\length$-length discrete-time signal $\signal$ (see Algorithm~\ref{alg:FFAST} in Appendix~\ref{app:pseudocode})
\begin{enumerate}
\item {\bf Sub-sampling:} A FFAST architecture (see Fig.~\ref{fig:conceptual} in Section~\ref{sec:intro}) with $\stages$ stages, has $\stages$ distinct subsampling patterns chosen as per the discussions in Sections \ref{sec:verysparse} and \ref{sec:lesssparse}. These patterns are deterministic and are pre-computed. We assume the presence of random-access-memory, with unit cost per I/O operation, for reading the subsamples. For the very-sparse regime of $ 0 < \sindex \leq 1/3$, as shown in Section~\ref{sec:verysparseconstruction}, the FFAST architecture with $\stages=3$ stages  is sufficient to compute a $\sparsity$-sparse $\length$-length DFT using $\samples = 2.44\sparsity$ samples. For the less-sparse regime of $ 1/3 < \sindex \leq 0.99$, as shown in Section~\ref{sec:lesssparseconstruction}, we have three different FFAST architectures, with $\stages = 6$ and $\stages = 8$, for different range of $\sindex$. The total samples used for these regions are $\samples = 3.14\sparsity$ and $\samples = 3.74\sparsity$ respectively. 

In general for any fixed $0 < \sindex < 1$, the sample complexity $\samples$ can be as small as $\osratio\sparsity$, where $\osratio$ is a constant that depends on the sparsity index $\sindex$.

\item {\bf DFT:} The FFAST algorithm then computes $2d$ DFT's, each of length approximately equal to $\binexp\sparsity$, for some constant $\binexp$ depending on the number of stages $\stages$ of the FFAST front-end. Using a standard FFT algorithm, e.g., prime-factor FFT or Winograd FFT algorithm \cite{blahut1985fast}, one can compute each of these DFT's in $O(\sparsity\log \sparsity)$ computations. Thus, the total computational cost of this step is $O(\sparsity\log \sparsity)$.

\item {\bf Peeling-decoder over sparse graph codes:} It is well known \cite{LMSS01}, that the computational complexity of the FFAST peeling-decoder over sparse graph codes is linear in the dimension of the graph, i.e., $O(\sparsity)$.
\end{enumerate}

Thus, the FFAST algorithm computes a $\sparsity$-sparse $\length$-length DFT with $O(\sparsity)$ samples using no more than $O(\sparsity\log \sparsity)$ computations for all $0 < \sindex < 1$.
\section{Simulation Results}\label{sec:experiments} In this section we validate the empirical performance of the FFAST algorithm for computing the DFT of signals having a sparse Fourier spectrum. We contrast the observed empirical performance, in terms of sample complexity, computational complexity and probability of success, with the theoretical claims of Theorem~\ref{thm:main}.\begin{figure}
\begin{center}
\includegraphics[scale = 0.5]{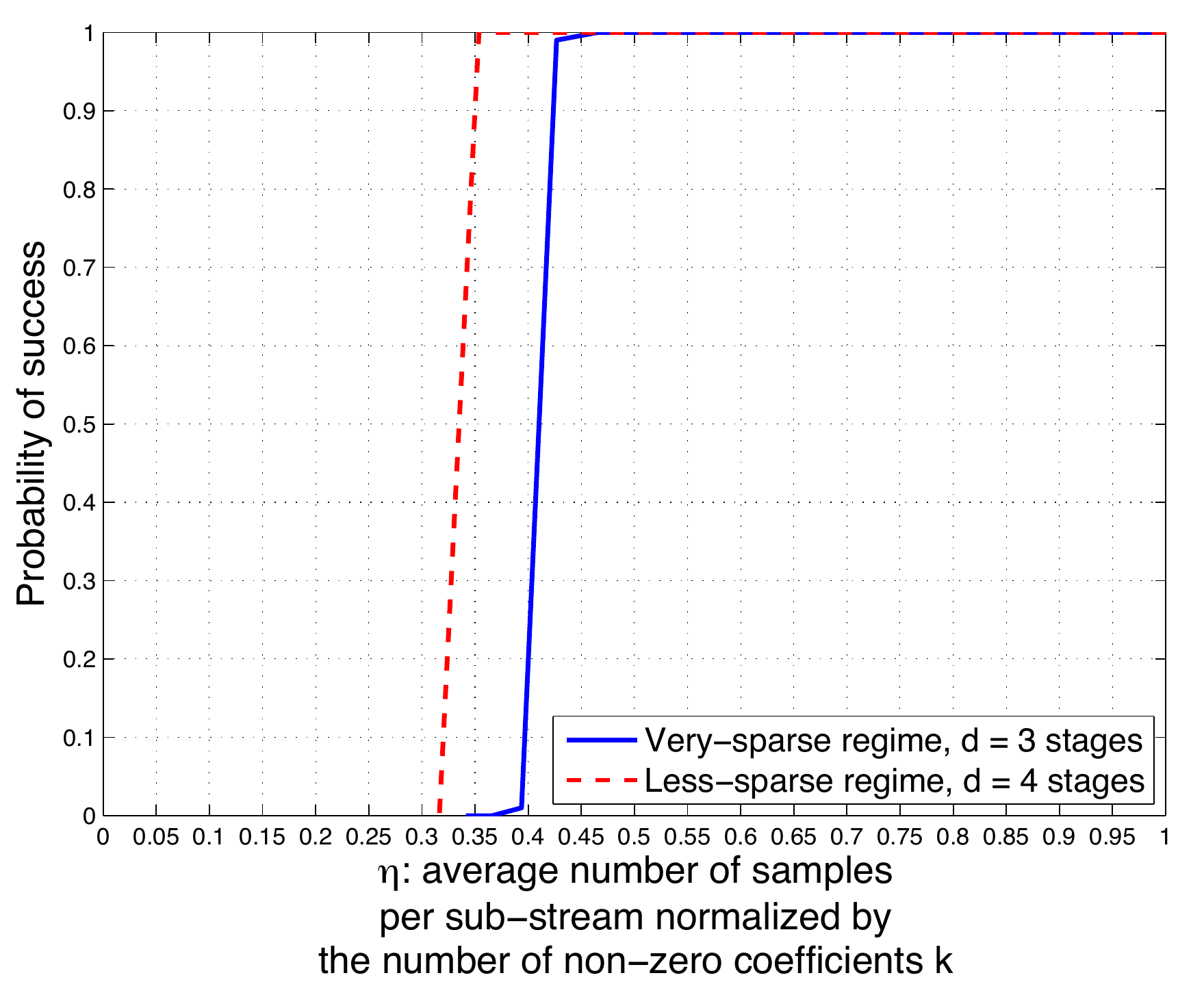}
\caption{The probability of success of the FFAST algorithm as a function of $\binexp$, the average number of samples per delay-chain, normalized by the number of non-zero DFT coefficients. The plot is obtained for two different sparsity regimes: 1) Very-sparse regime: $\sparsity =O(\length^{1/3})$. For this regime, a $\stages=3$ stage FFAST architecture is chosen. The number of samples per sub-stream in each of the $3$ stages are perfectly co-prime: $\factor{0} = 511, \factor{1} = 512$ and $\factor{2} = 513$ respectively, and 2) Less-sparse regime: $\length^{0.73} < \sparsity < \length^{0.85}$. For this regime, a $\stages=4$ stage FFAST architecture is chosen. The number of samples per sub-stream in each of the $4$ stages are {\em not} co-prime but have ``cyclically-shifted" overlapping co-prime factors: $\factor{0} = 16\times17\times19 = 5168$, $\factor{1} = 17\times19\times21 = 6783$, $\factor{2} = 19\times21\times16 = 6384$ and $f_3 = 21\times16\times17 = 5712$ respectively. Each point on the plot is obtained by averaging over 10000 runs. The ambient signal dimension $\length$ and the number of samples $\samples$ are fixed, while the number of non-zero coefficients $\sparsity$ is varied to get different values of $\binexp$. We note that the FFAST algorithm exhibits a threshold behavior with $\binexp_1 = 0.427$ being the threshold for the very-sparse regime, and $\binexp_2 = 0.354$ for the less-sparse regime respectively. From Table~\ref{tab:convergence} in Section~\ref{sec:cyclefree}, we see that the optimal threshold values are $\binexp^*_1 = 0.4073$ and $\binexp^*_2 = 0.3237$, which are in close agreement with our simulation results.}
\label{fig:threshold}
\end{center}
\end{figure}
 \begin{table}[h]
\begin{center}
\begin{tabular}{|c|c|c|c|c|c|c|c|c|}
  \hline
 \multirow{2}{*}{Regime}&  Signal length &  Sparsity & $\sindex$ & stages & \multicolumn{2}{c|}{$\samples \approx 2\stages\binexp \sparsity$} &  \multirow{2}{*}{$\ell$} &\multirow{2}{*}{failures}\\
 \cline{6-7} 
   & $\length$ & $\sparsity$ & & $\stages$ & \samples & $\binexp$ & &\\
  \hline
  \multirow{4}{*}{Very-sparse} & \multirow{2}{*}{511*512*513} & $900$ & $0.363$ & \multirow{4}{*}{3} & $3072$&$0.569$ & 6 &1/$10^4$\\
  \cline{3-4} \cline{6-9}
   &  & $1000$ & $0.369$ && $3072$&$0.512$ & 9 & 0/$10^4$\\
  \cline{3-4} \cline{6-9}
   & \multirow{2}{*}{$\approx 134 \times 10^6$}& $1100$ & $0.374$ & & $3072$&$0.465$ & 13 &1/$10^4$\\
  \cline{3-4} \cline{6-9}
  &  & $1200$ & $0.378$ && $3072$&$0.427$ & 18 &99/$10^4$\\
\hline
  \multirow{4}{*}{Less-sparse} & \multirow{2}{*}{16*17*19*21} & $13000$ & $0.81$ & \multirow{4}{*}{4} & $48094$&$0.462$ & 6&0/$10^4$\\
  \cline{3-4} \cline{6-9}
   &  & $15000$ & $0.83$ && $48094$&$0.401$ & 8 &0/$10^4$\\
  \cline{3-4} \cline{6-9}
   & \multirow{2}{*}{$\approx0.1 \times 10^6$}& $17000$ & $0.84$ & & $48094$&$0.354$ & 13 &2/$10^4$\\
  \cline{3-4} \cline{6-9}
  &  & $19000$ & $0.85$ && $48094$&$0.316$ & 29 &$10^4$/$10^4$\\
\hline
\end{tabular}
\end{center}
\caption{Performance of the FFAST algorithm for two different sparsity regimes: 1) Very-sparse regime: $\sparsity =O(\length^{1/3})$. For this regime, a $\stages=3$ stage FFAST architecture is chosen. The number of samples per sub-stream in each of the $3$ stages are perfectly co-prime: $\factor{0} = 511, \factor{1} = 512$ and $\factor{2} = 513$ respectively, and 2) Less-sparse regime: $\length^{0.73} < \sparsity < \length^{0.85}$. For this regime, a $\stages=4$ stage FFAST architecture is chosen. The number of samples per sub-stream in each of the $4$ stages are {\em not} co-prime but have ``cyclically-shifted" overlapping co-prime factors: $\factor{0} = 16\times17\times19 = 5168$, $\factor{1} = 17\times19\times21 = 6783$, $\factor{2} = 19\times21\times16 = 6384$ and $f_3 = 21\times16\times17 = 5712$ respectively. We note that when $\binexp_1 \geq 0.427$ and $\binexp_2 \geq 0.354$, for the very-sparse and the less-sparse regimes respectively, the FFAST algorithm successfully computes all the non-zero DFT coefficients for almost all the runs. Further, in one or two instances when it failed to recover all the non-zero DFT coefficients, it has recovered all but $8$ (for $\stages=3$) or $16$ (for $\stages=4$) non-zero DFT coefficients. This validates our theoretical findings. From Table~\ref{tab:convergence} in Section~\ref{sec:cyclefree}, we see that the optimal threshold values for the very-sparse and less-sparse regimes are $\binexp^*_1 = 0.4073$ and $\binexp^*_2 = 0.3237$ resp., which are in close agreement with the simulation results. The table also shows that the average number of iterations $\ell$, required for the FFAST algorithm to successfully compute the DFT for both the regimes, are quite modest.}\label{tab:probsuccess}
\end{table}

\subsection{FFAST architecture and signal generation}
\begin{itemize}
\item {\bf Very sparse regime:} The input signal $\signal$ is of length $\length = 511\times512\times513$ $\approx 134 \times 10^6$. The number of non-zero DFT coefficients $\sparsity$ is varied from $400$ to $1300$, this corresponds to the very-sparse regime of $\sparsity =O(\length^{1/3})$. We use an FFAST architecture with $\stages=3$ stages and $2$ delay-chains per stage. The uniform sampling periods for the $3$ stages are $512\times513$, $511\times513$ and $511\times512$ respectively. This results in the number of output samples, per delay-chain, for the three stages to be $\factor{0} = 511$, $\factor{1} = 512$ and $\factor{2} = 513$ respectively. Note that the number of samples in the three different stages, i.e., $\factor{i}$'s, are pairwise co-prime. The total number of samples used by the FFAST algorithm for this simulation is\footnote{The samples used by the different sub-streams in the three different stages overlap, e.g., $\signalc{0}$ is common to all the zero delay sub-streams in each stage. Hence, $\samples < 2(\factor{0}+\factor{1}+\factor{2})$ and not equal.} $\samples < 2(\factor{0}+\factor{1}+\factor{2}) = 3072$.

\item {\bf Less sparse regime:} The input signal $\signal$ is of length $\length = 16\times17\times19\times21$ $\approx0.1 \times 10^6$. The number of non-zero DFT coefficients $\sparsity$ is varied from $5000$ to $19000$, this corresponds to the less-sparse regime of $\length^{0.73}< \sparsity < \length^{0.85}$. We use an FFAST architecture with $\stages=4$ stages. The uniform sampling periods for the $4$ stages are $21$, $16$, $17$ and $19$ respectively. This results in the number of output samples, per delay-chain, for the four stages to be $\factor{0} = 16\times17\times19 = 5168$, $\factor{1} = 17\times19\times21 = 6783$, $\factor{2} = 19\times21\times16 = 6384$ and $\factor{3} = 21\times16\times17 = 5712$ respectively. Note that the number of samples in the four stages are composed of ``cyclically-shifted" co-prime numbers and are not pairwise co-prime. The total number of samples used by the FFAST algorithm for this simulation is $\samples < 2(\factor{0}+\factor{1}+\factor{2}+f_3) = 48094$.

\item For each run, an $\length$-dimensional $\sparsity$-sparse signal $\transform$ is generated with non-zero values $X_i \in \{\pm 10\}$ at uniformly random  positions in $\{0,1\hdots,\length-1\}$. The time domain signal $\signal$ is then generated from $\transform$ using an IDFT operation. This discrete-time signal $\signal$ is then given as an input to FFAST.
\item Each sample point in the Fig.~\ref{fig:threshold}, is generated by averaging over 10000 runs.
\item Decoding is successful if all the DFT coefficients are recovered perfectly.
 \end{itemize}
 
\subsection{Observations and conclusions} We observe the following aspects of the simulations in detail and contrast it with the claims of Theorem~\ref{thm:main}
 
\paragraph*{\bf Density Evolution} The density evolution recursion \eqref{eq:evolution2} in Section~\ref{sec:cyclefree}, implies a threshold behavior: if the average number of samples per delay-chain normalized by $\sparsity$, i.e., $\binexp$, is above a certain threshold (as specified in Section~\ref{sec:cyclefree} Table~\ref{tab:convergence}), then the recursion \eqref{eq:evolution2} converges to $0$ as $\ell \rightarrow \infty$ otherwise $p_{\ell}$ is strictly bounded away from $0$. In Fig.~\ref{fig:threshold}, we plot the probability of success of the FFAST algorithm averaged over 10000 independent runs, as a function of $\binexp$, i.e., varying $\sparsity$ for a fixed number of samples $\samples$. We note that the FFAST algorithm also exhibits a threshold behavior with $\binexp_1 = 0.427$, being the threshold for the very-sparse regime with $\stages=3$, and $\binexp_2 = 0.354$ for the less-sparse regime with $\stages=4$ respectively. From Table~\ref{tab:convergence} in Section~\ref{sec:cyclefree}, we see that the optimal threshold values are $\binexp^*_1 = 0.4073$ and $\binexp^*_2 = 0.3237$, which are in close agreement with our simulation results of Fig.~\ref{fig:threshold}.

\paragraph*{\bf Iterations} In the theoretical analysis of Section~\ref{sec:convergencetocyclefreecase}, we have shown that the FFAST algorithm, if successful, decodes all the DFT coefficients in $\ell$ iterations, where, $\ell$ is a large but a fixed constant. In Table~\ref{tab:probsuccess}, we observe that when successful, the FFAST peeling-decoder completes decoding in few ($\leq13$ in this case) number of iterations.

 \paragraph*{\bf Probability of success} The empirical probability of success of the FFAST algorithm for the very-sparse as well as the less-sparse regimes is shown in Table~\ref{tab:probsuccess}. We observe, as long as the parameter $\binexp$ is above the minimum threshold $\binexp^*$ dictated by Table~\ref{tab:convergence}, in Section~\ref{sec:cyclefree}, the FFAST algorithm indeed computes all the non-zero DFT coefficients successfully with high probability. Further, in one or two instances when it failed to recover all the non-zero DFT coefficients, it has recovered all but $8$ (for $\stages=3$) or $16$ (for $\stages=4$) non-zero DFT coefficients. Thus, confirming our theoretical analysis of the worst case error event, of being trapped in a cycle of size $2^{\stages}$.

\section{Conclusion}\label{sec:conclusion}
In this paper, we addressed the problem of computing an $\length$-length DFT of signals that have $\sparsity$-sparse Fourier transform, where $\sparsity << \length$. We proposed a novel FFAST algorithm that cleverly exploits filterless subsampling operation to induce aliasing artifacts, similar to parity-check constraints of good erasure-correcting sparse-graph codes, on the spectral coefficients. Then, we formally connected the problem of computing sparse DFT to that of decoding of appropriate sparse-graph codes. This connection was further exploited to design a sub-linear complexity FFAST peeling-style back-end decoder. Further, we analyzed the performance of the FFAST algorithm, using well known density evolution techniques from coding theory, to show that our proposed algorithm computes the $\sparsity$-sparse $\length$-length DFT using only $O(\sparsity)$ samples in $O(\sparsity\log\sparsity)$ arithmetic operations, with high probability. The constants in the big Oh notation for both sample and computational cost are small, e.g., when $\sindex < 0.99$, which essentially covers almost all practical cases of interest, the sample cost is less than $4\sparsity$. We also provide simulation results, that are in tight agreement with our theoretical findings.

Although, the focus of this paper is on signals having an exactly-sparse DFT, the key insights derived from analyzing the exactly sparse signal model apply more broadly to the noisy setting, i.e., where the observations are further corrupted by noise \cite{pawar2014isit}. In particular, in \cite{pawar2014isit}, we show that the robustness against sample noise is achieved by an elegant modification of the noiseless FFAST framework. Specifically, the noiseless FFAST framework shown in Fig.~\ref{fig:conceptual} has multiple sub-sampling stages, where each stage has 2 Òdelay-chainsÓ. In contrast, for the case of noise robust FFAST framework, we use $O(\log\length)$ number of delay-chains per sub-sampling stage, further the input signal to each delay-chain is circularly shifted by a random amount prior to sub-sampling, i.e., random delays. In \cite{pawar2014isit}, we show that a random choice of the circular shifts endows the effective ``measurement matrix" with a good mutual incoherence and Restricted-Isometry-Property \cite{candes2005decoding}, thus resulting in a stable recovery.
 
\begin{figure}[h]
\begin{center}
\includegraphics[width = 0.8\linewidth]{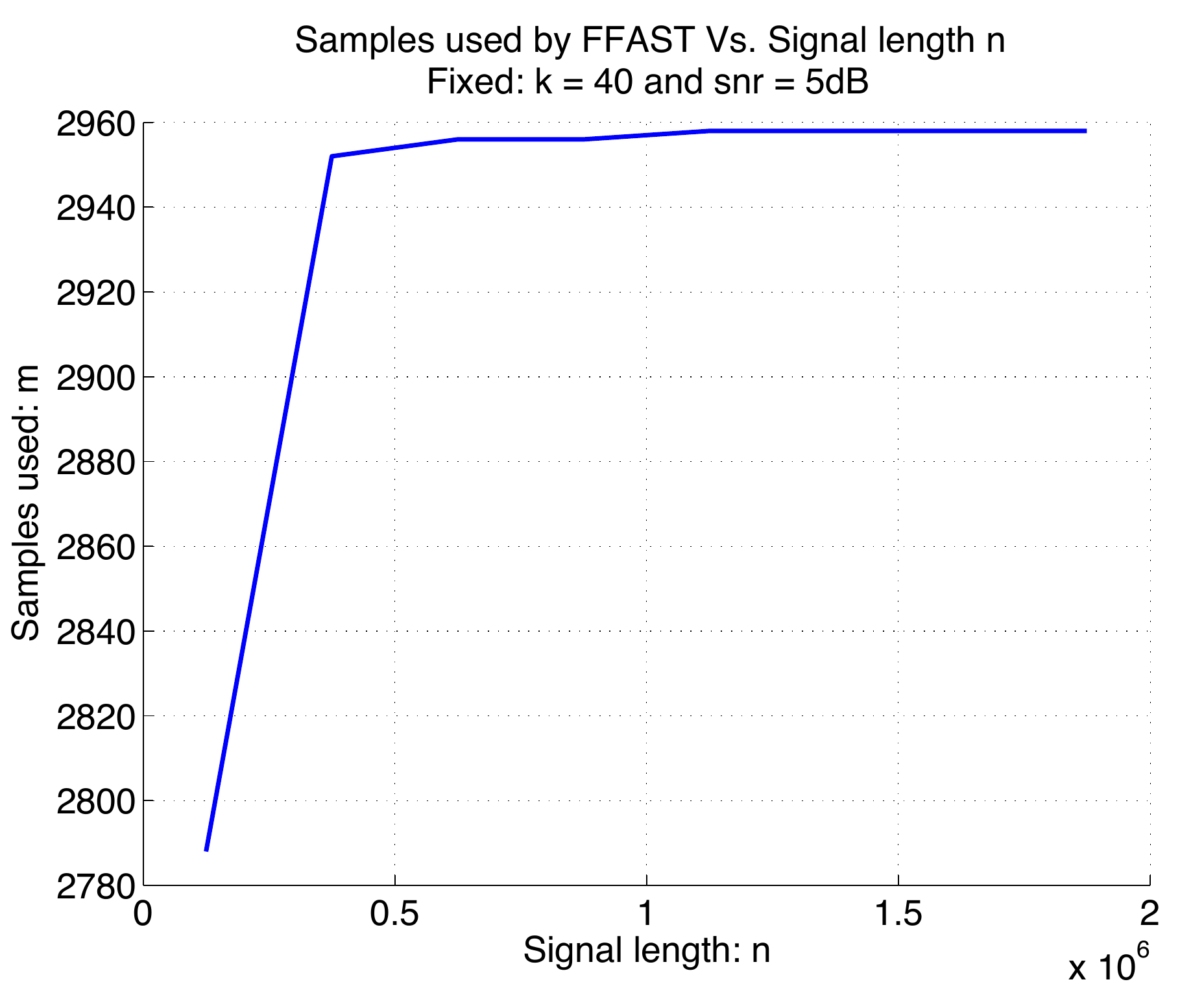}
\caption{The plot shows the scaling of the number of samples $\samples$ required by the noise robust FFAST algorithm to successfully reconstruct a $\sparsity=40$ sparse DFT $\transform$, for increasing signal length $\length$. For a fixed support recovery probability of $0.99$, sparsity $\sparsity$ and SNR of $5dB$, we note that $\samples$ scales logarithmically with increasing $\length$.}
\label{fig:random}
\end{center}
\end{figure}

As an example, we provide here an empirical evaluation of the scaling of the number of samples $\samples$, required by a noise robust FFAST algorithm \cite{pawar2014isit}, to successfully compute the DFT $\transform$ from noise-corrupted samples, as a function of the signal length $\length$. A noise-corrupted signal $\obs = \signal + \noise$, where $\noise$ is Gaussian, of length $\length$ is generated such that the effective signal-to-noise ratio is $5$dB and the DFT of $\signal$ has $\sparsity = 40$ non-zero coefficients. The signal $\length$ is varied from $\length = 49*50*51 \approx 0.1$ million to $\length = 49*50*51*15 \approx 1.87$ million. The noise-corrupted samples $\obs$ are processed using a $\stages =3$ stage FFAST architecture with $\delays$ delay-chains per stage. As the signal length $\length$ increases the number of delays $\delays$ per stage are increases to achieve at least $0.99$ probability of successful support recovery. From the plot in Fig.~\ref{fig:random}, we note that $\samples$ scales logarithmically with increasing $\length$. For further details about this simulation see \cite{pawar2014isit}.

The FFAST framework in this paper, has inspired the works of \cite{scheibler2013fast} and \cite{simon2014isit}, where the authors have proposed novel algorithms for computing high dimensional sparse ``Walsh Hadamard Transform" (WHT), for the noiseless and noisy observed samples respectively. The framework has also been successfully applied for blind acquisition of signals with sparse spectra using sub-Nyquist sampling \cite{pawar2013GlobalSIP} and to a variant of a phase-retrieval-problem as in \cite{pedarsani2014phasecode}. We are hopeful that many more applications will follow.
 \pagebreak
\begin{appendix}
\subsection{FFAST Algorithm Pseudocode}\label{app:pseudocode}
\begin{algorithm}[h!]
\caption{FFAST Algorithm}
\label{alg:FFAST}
\begin{algorithmic}[1]
\STATE {\em Inputs:} \\
- A discrete time signal $\signal$ of length $\length$, whose $\length$-length DFT $\transform$ has at most $\sparsity$ non-zero coefficients.\\
- The subsampling parameters of the FFAST architecture (see Fig.~\ref{fig:conceptual}): 1) number of stages $\stages$. 
		and 2) number of samples per sub-stream in each of the $\stages$ stages $\setfactors = \{\factor{0},\factor{1},\hdots,\factor{\stages-1}\}$, chosen as per discussion in Sections \ref{sec:verysparse} and \ref{sec:lesssparse}.
\algrule
\STATE {\em Output:}  An estimate of the $\sparsity$-sparse $\length$-length DFT $\transform$.
\algrule
\STATE {\em FFAST Decoding:} Set the initial estimate of the $\length$-length DFT $\transform = 0$. Let $\iterations$ denote the number of iterations performed by the FFAST decoder.
\FOR {each iteration}
	\FOR {each stage $i =0$ to $\stages-1$}
		\FOR {each bin $j=0$ to $\factor{i}-1$}
			\IF {$||\binobsv{i}{j}||^2 == 0$}
                		\STATE bin $j$ of stage $i$ is a zero-ton.
            		\ELSE
				\STATE (single-ton, $v_p$, $p$) = {\bf \emph{Singleton-Estimator}} $(\binobsv{i}{j})$.
	       			\IF {single-ton = `true'}
                        		\STATE {\em Peeling}: $\binobsv{s}{q} = \binobsv{s}{q} - \
v_p\left(
\begin{array}{c}
  1\\
  e^{\imath2\pi p/\length} 
\end{array}
\right)$, for all stages $s$ and bins $q \equiv p \text{ mod } \factor{s}$.
                      		\STATE Set $\transformc{p} = v_p$.
				\ELSE
					\STATE bin $j$ of stage $i$ is a {\em multi-ton}.
				\ENDIF
			\ENDIF
		\ENDFOR
	\ENDFOR
\ENDFOR
\end{algorithmic}
\end{algorithm}

\begin{algorithm}[h!]
\caption{ Singleton-Estimator}
\label{alg:single-ton_estimator}
\begin{algorithmic}[1]
\STATE {\em Input:} The bin observation $\binobsv{i}{j}$.
\algrule
\STATE {\em Outputs:} 1) A boolean flag `single-ton', 2) Estimated value $v_p$ of the non-zero DFT coefficient at position $p$.
\algrule
\STATE {\em Singleton-Estimator:} Set the single-ton = `false'.
\IF {$(\length/2\pi)\angle \binobsc{i}{j}[1]\conj{\binobsc{i}{j}[0]} \in \{0,1,\hdots,(\length-1)\}$}
	\STATE single-ton = `true'.
        \STATE $v_p = \binobsc{i}{j}[0]$.
        \STATE $p = (\length/2\pi)\angle \binobsc{i}{j}[1]\conj{\binobsc{i}{j}[0]}$.
  \ENDIF
\end{algorithmic}
\end{algorithm}
\pagebreak
\subsection{Ratio-test for a multi-ton bin}\label{app:multi-tonratiotest} 
Consider a multi-ton-bin $\ell$ with $\mlength-1$ non-zero components where $\mlength > 2$. Let $i_0,i_1,\cdots, i_{\mlength-2}$ be the indices of the non-zero components of $\transform$ contributing to the observation of the multi-ton bin $\ell$. Then, we have 
\begin{eqnarray}
\binobsvi{\ell} &=&
\left[
\begin{array}{ccccc}
  \vec{v}_{i_0}& \vec{v}_{i_1}  & \cdots & \vec{v}_{i_{\mlength-3}} & \vec{v}_{i_{\mlength-2}}  
\end{array}
\right] \left[
\begin{array}{c}
  \transformc{i_0}\\
  \transformc{i_1}\\
  \vdots\\
  \transformc{i_{\mlength-3}}\\
  \transformc{i_{\mlength-2}}
\end{array}
\right],
\end{eqnarray}
where vector $\vec{v}_{i_0} = \left(
\begin{array}{c}
  1 \\
  e^{\imath2\pi i_o/\length}
\end{array}
\right)$ consists of the first two entries of the $i_0^{th}$ column of the $\length\times\length$ IDFT matrix. The multi-ton bin $\ell$ will be identified as a single-ton bin with location $j$ and value $\transformc{j}$ for some $0 \leq j < \length$ iff 
\begin{eqnarray}\label{eq:multi-toncondn}
\binobsvi{\ell} &=& \vec{v}_{j} \transformc{j} \nonumber\\
i.e., \left[
\begin{array}{cccc}
  \vec{v}_{i_0}& \vec{v}_{i_1}  & \cdots & \vec{v}_{i_{\mlength-3}}
\end{array}
\right] \left[
\begin{array}{c}
  \transformc{i_0}\\
  \transformc{i_1}\\
  \vdots\\
  \transformc{i_{\mlength-3}}
\end{array}
\right] &=& \left[
\begin{array}{cc}
 \vec{v}_{i_{\mlength-2}} & \vec{v}_{j}
\end{array}
\right] \left[
\begin{array}{c}
  -\transformc{i_{\mlength-2}}\\
  \transformc{j}
\end{array}
\right] \nonumber\\
\vec{u} &=& \left[
\begin{array}{cc}
 \vec{v}_{i_{\mlength-2}} & \vec{v}_{j}
\end{array}
\right] \left[
\begin{array}{c}
  -\transformc{i_{\mlength-2}}\\
  \transformc{j}
\end{array}
\right].
\end{eqnarray}
where $\vec{u} \in \complex^{2}$ is some resultant vector. The matrix consisting of the first two rows of an IDFT matrix is a Vandermonde matrix and hence equation~\eqref{eq:multi-toncondn} has a unique solution. The complex coefficient $\transformc{j}$ is free to choose but $\transformc{i_{\mlength-2}}$ is drawn from a continuous distribution. As a result probability of satisfying equation~\eqref{eq:multi-toncondn} is essentially zero even after applying union bound over all $j$. Hence a multi-ton bin is identified correctly almost surely.

\subsection{Check node edge degree-distribution polynomial of the graphs in the balls-and-bins ensemble $\bbensemble$}\label{app:degdistribution}
The check node edge-degree distribution polynomial of a bi-partite graph is defined as $\rho(\alpha) = \sum_{i=1}^\infty \rho_i \alpha^{i-1}$, where $\rho_i$ denotes the probability of an edge (or fraction of edges) of the graph is connected to a check node of degree $i$. Recall that in the randomized construction based on a balls-and-bins model described in Section~\ref{sec:bbensemble}, every variable node is connected to $\stages$ check nodes, one uniformly random check node in each of the $\stages$ sets. Thus, the number of edges connected to a check node, in the subset with $\factor{0}$ check nodes, is a binomial $B(1/\factor{0},\sparsity)$ random variable. For large $\sparsity$ and $\factor{0} = \binexp\sparsity + O(1)$, the binomial distribution $B(1/\factor{0},\sparsity)$ is well approximated by a Poisson random variable with mean $1/\binexp$. Thus, 
\begin{equation}\label{eq:checknodedgree}
Pr(\text{check node in set $0$ has edge degree} = i) \approx \frac{(1/\binexp)^i e^{-1/\binexp}}{i!}.
\end{equation}
Let $\rho_{0,i}$ be the fraction of the edges, that are connected to a check node of degree $i$ in set $0$. Then, we have,
\begin{eqnarray}
\rho_{0,i} &=& \frac{i \factor{0} }{\sparsity}Pr(\text{check node has edge degree} = i)\nonumber\\
&\overset{(a)}{\approx}& \frac{i \factor{0}}{\sparsity}\frac{(1/\binexp)^i e^{-1/\binexp}}{i!}\nonumber\\
&\overset{(b)}{\approx}& \frac{(1/\binexp)^{i-1}e^{-1/\binexp}}{(i-1)!},
\end{eqnarray}
where $(a)$ follows from Poisson approximation of the Binomial random variable and $(b)$ from $\factor{0} = \binexp \sparsity + O(1)$. Since, $\factor{i} = \binexp \sparsity + O(1)$ for all $i=0,1\hdots,\stages-1$. We have,
\begin{eqnarray}
\rho_{i} &\approx& \frac{(1/\binexp)^{i-1}e^{-1/\binexp}}{(i-1)!}
\end{eqnarray}
and 
\begin{eqnarray}\label{eq:degreedistribution}
\rho(\alpha) = e^{-(1-\alpha)/\binexp}
\end{eqnarray}

\subsection{Proof of Lemma~\ref{lem:expsucceeds}}\label{app:expsucceeds}
\begin{proof}
We provide a proof using a contradiction argument. If possible let $S$ be the set of the variable nodes that the FFAST peeling-decoder fails to decode. We have $|S| \leq \alpha \sparsity$. Without loss of generality let $|N_1(S)| \geq |N_i(S)|, \ \forall i \in \{0,\hdots,\stages-1\}$. Then, by the hypothesis of the Theorem $|N_1(S)| > |S|/2$. 

Note that the FFAST peeling-decoder fails to decode the set $S$ if and only if there are no more single-ton check nodes in the neighborhood of $S$ and in particular in $N_1(S)$. For all the check nodes in $N_1(S)$ to be a multi-ton, the number of edges connecting to the check nodes in the set $N_1(S)$ have to be at least $2|N_1(S)| > |S|$. This is a contradiction since there are only $|S|$ edges going from set $S$ to $N_1(S)$ by construction.
\end{proof}

\subsection{Proof of Lemma~\ref{lem:ballsandbinsexp}}\label{app:expander_ballsandbins}
\begin{proof}
Consider a set $S$ of variable (left) nodes in a random graph from the ensemble $\bbensemble$, where $|{\setfactors}| = \stages \geq 3$. Let $N_i(S)$ be the right neighborhood of the set $S$ in the $i^{th}$ subset of check nodes, for $i=0,1,\hdots,\stages-1$. Also, let $E_S$ denote the event that the all the $\stages$ right neighborhoods of $S$ are of size $|S|/2$ or less, i.e., $\max \{|N_i(S)|\}_{i=0}^{\stages-1} \leq |S|/2$. First, we compute an upper bound on the probability of the event $E_S$ as follows:
\begin{eqnarray}
Pr(E_S) &<& \prod_{i=0}^{\stages-1} \left(\frac{|S|}{2f_i}\right)^{|S|} {\factor{i} \choose |S|/2}\nonumber\\
&\overset{(a)}{\approx}& \left(\frac{|S|}{2\binsize}\right)^{\stages|S|} {\binsize \choose |S|/2}^{\stages}\nonumber\\
&<& \left(\frac{|S|}{2\binsize}\right)^{\stages|S|} \left(\frac{2\binsize e}{|S|}\right)^{\stages|S|/2}\nonumber\\
&=& \left(\frac{|S|e}{2\binsize}\right)^{\stages|S|/2}
\end{eqnarray}
where the approximation $(a)$ uses $\factor{i} = \binsize + O(1)$ for all $i=0,\hdots,\stages-1$. Consider an event $E$, that there exists some set of variable nodes of size $|S|$, whose all the $\stages$ right neighborhoods are of size $|S|/2$ or less. Applying a union bound we get,
\begin{eqnarray}\label{eq:expander_bound1}
Pr(E) &<& Pr(E_S) {\sparsity \choose |S|}\nonumber\\
&<& \left(\frac{|S|e}{2\binsize}\right)^{\stages|S|/2} \left(\frac{\sparsity e}{|S|}\right)^{|S|}\nonumber\\
&\overset{(b)}{<}&\left[ \left(\frac{|S|}{\binsize}\right)^{\stages-2} \left(\frac{e}{2}\right)^\stages \left(\frac{e}{\binexp}\right)^2 \right]^{|S|/2}\nonumber\\
&<&O\left( (|S|/\nbins)^{|S|/2}\right)
\end{eqnarray}
where in $(b)$ we used $\binsize = \binexp \sparsity$ and in the last inequality we have used $\stages\geq 3$ and $\nbins = O(\binsize)$. Then, specializing the bound in \eqref{eq:expander_bound1} for $|S| = \alpha \sparsity$, for small enough $\alpha >0$, and $|S| = o(\sparsity)$ we get,
\begin{itemize}
\item For $|S| = \alpha \sparsity$, for small enough $\alpha > 0$:
\begin{eqnarray}
Pr(E) &<& e^{-\epsilon \sparsity \log(\nbins/\sparsity)}, \ \ \text{for some} \ \ \epsilon > 0
\end{eqnarray}
\item For $|S| = o(\sparsity)$:
\begin{eqnarray}
Pr(E) &<& O\left(1/\nbins\right)
\end{eqnarray}
\end{itemize}
\end{proof}

\subsection{Proof of Lemma~\ref{lem:concentration}}\label{app:proofconcentration}
\begin{proof} a) [Expected behavior] Consider decoding on a random graph from the ensemble $\bbensemble$. Let $Z_i, \ i = 0,\hdots,\sparsity\stages-1$, be an indicator random variable that takes value $1$ if the edge $\vec{e}_i$ is not decoded after $\ell$ iterations of the FFAST peeling-decoder and $0$ otherwise. Then, by symmetry $\expectation[Z] = \sparsity\stages\expectation[Z_1]$. Next, we compute $\expectation[Z_1]$ as,
\begin{eqnarray*}
\expectation[Z_1] &=& \expectation[Z_1 \mid {\cal N}^{2\ell}_{\vec{e}_1} \text{ is tree-like}] Pr({\cal N}^{2\ell}_{\vec{e}_1} \text{ is tree-like}) \\
&&+ \expectation[Z_1 \mid {\cal N}^{2\ell}_{\vec{e}_1} \text{ is not tree-like}] Pr({\cal N}^{2\ell}_{\vec{e}_1} \text{ is not tree-like}) \\
&\leq& \expectation[Z_1 \mid {\cal N}^{2\ell}_{\vec{e}_1} \text{ is tree-like}] + Pr({\cal N}^{2\ell}_{\vec{e}_1} \text{ is not tree-like}).
\end{eqnarray*}
In Appendix~\ref{app:treelike}, we show that $Pr({\cal N}^{2\ell}_{\vec{e}_1} \text{ is not tree-like}) \leq O((\log \sparsity)^{2\ell}/\sparsity)$. Also we have $\expectation[Z_1 \mid {\cal N}^{2\ell}_{\vec{e}_1} \text{is tree-like}]  = p_{\ell}$ by definition. Thus,
\begin{equation}
\expectation[Z] < 2kdp_{\ell}.
\end{equation}

b) [Concentration] In this part, we want to show that the number of edges $Z$, that are not decoded at the end of $\ell$ iterations of the FFAST peeling-decoder, is highly concentrated around $\expectation[Z]$. We use the standard martingale argument, of exposing an edge of the graph at a time, along with the Azuma's inequality as in \cite{richardson2001capacity}, with slight modification to account for the irregular edge degree of the right nodes. The bi-partite graphs in the ensemble $\bbensemble$, unlike the ones in \cite{richardson2001capacity}, have left nodes with $\stages$-regular edge degree,s while the right nodes edge degree is a Poisson random variable with a constant mean. In particular, let $Y_i = \expectation[Z|\vec{e}_i]$ be the expected value of $Z$ after exposing $i$ edges. Then, $Y_0 = \expectation[Z]$ and $Y_{\sparsity\stages} = Z$, and the sequence $Y_i$ forms a Doob martingale. If $|Y_{i} - Y_{i-1}| < c_i$ then, using Azuma's inequality we have for all reals $t$,
\begin{equation}\label{eq:azuma}
\pr(|Z - \expectation[Z]| > t) \leq 2\exp\left(\frac{-t^2}{2\sum_{i=1}^{\sparsity\stages}c^2_i}\right).
\end{equation}
In \cite{richardson2001capacity}, the authors have shown that for bi-partite graphs with edge degrees $d_v$ for left nodes and $d_c$ for right nodes, $c_i < 8(d_cd_v)^{\ell}$, when the peeling-decoder performs message passing over $\ell$ iterations. The graphs considered in this paper have constant left degree, but the edge degree of the right nodes is a {\em Poisson random variable} with a constant mean, hence we probabilistically bound the value of $d_c$. Let $B$ be an event that all the right nodes of a random graph from the considered ensemble have edge degrees less than or equal to $O(\sparsity^{1/(2\ell+0.5)})$. Then, using a bound on the Poisson distribution and a union bound over all the $O(\sparsity)$ check nodes, we get,
\begin{eqnarray}\label{eq:nodeoutage}
\pr(\bar{B}) < O(\sparsity e^{-\beta_1\sparsity^{\frac{1}{2\ell+0.5}}}).
\end{eqnarray}
Conditioned on the event $B$, one can upper bound $c_i$ by $O(\sparsity^{\ell/(2\ell+0.5)})$.
Now putting \eqref{eq:azuma} and \eqref{eq:nodeoutage} together, we have,
\begin{eqnarray}
\pr(|Z - \expectation[Z]| > \sparsity\stages\epsilon_1) &\leq& \pr(|Z - \expectation[Z]| > \sparsity\stages\epsilon_1 | B) +  \pr(\bar{B})\nn \\
&\leq& 2\exp\left(\frac{-\sparsity^2\stages^2\epsilon_1^2}{2\sum_{i=1}^{\sparsity\stages}c^2_i}\right) + O(\sparsity e^{-\beta_1\sparsity^{\frac{1}{2\ell+0.5}}})\nn\\
&\leq& e^{-\beta\epsilon_1^2\sparsity^{1/(4\ell+1)}}.
\end{eqnarray}
\end{proof}

\subsection{Probability of Tree-like Neighborhood}\label{app:treelike}
Consider an edge $\vec{e}$ in a randomly chosen graph ${\cal G} \in \bbensemble$. Next, we show that a $2\ell^*$ depth neighborhood ${\cal N}^{2\ell^*}_{\vec{e}}$ of the edge $\vec{e}$ is tree-like with high probability, for any fixed $\ell^*$. Towards that end, we first assume that the neighborhood ${\cal N}^{2\ell}_{\vec{e}}$ of depth $2\ell$, where $\ell < \ell^*$, is tree-like and show that it remains to be tree-like when extended to depth $2\ell+1$ with probability approaching $1$. Let $C_{\ell,i}$ be the number of check nodes, from set $i$, and $M_{\ell}$ be the number of variable nodes, present in ${\cal N}^{2\ell}_{\vec{e}}$. Also assume that $t$ more edges from the leaf variable nodes in ${\cal N}^{2\ell}_{\vec{e}}$ to the check nodes in set $i$ at depth $2\ell+1$ are revealed without creating a loop. Then, the probability that the next revealed edge from a leaf variable node to a check node (say in set $i$) does not create a loop is $\frac{\factor{i} - C_{\ell,i} - t}{\factor{i} - C_{\ell,i}} \geq 1 - \frac{C_{\ell^*,i}}{\factor{i}-C_{\ell^*,i}}$. Thus, the probability that ${\cal N}^{2\ell+1}_{\vec{e}}$ is tree-like, given ${\cal N}^{2\ell}_{\vec{e}}$ is tree-like, is lower bounded by $\min_i (1 - \frac{C_{\ell^*,i}}{\factor{i}-C_{\ell^*,i}})^{C_{\ell+1,i} - C_{\ell,i}}$. Similarly assume that ${\cal N}^{2\ell+1}_{\vec{e}}$ is tree-like and $s$ more edges from check nodes to the variable nodes at depth $2\ell+2$ are revealed without creating a loop. Conditioned on the event that a check node has an outgoing edge it has equal chance of connecting to any of the edges of the variable nodes that are not yet connected to any check node. Thus, the probability of revealing a loop creating edge from a check node to a variable node at depth $2\ell+2$ is upper bounded by, $(1 - \frac{(\sparsity- M_{\ell}-s)\stages}{\sparsity\stages- M_{\ell} \stages - s}) \leq \frac{M_{\ell^*}}{(\sparsity-M_{\ell^*})}$. Thus, the probability that ${\cal N}^{2\ell+2}_{\vec{e}}$ is tree-like given ${\cal N}^{2\ell+1}_{\vec{e}}$ is tree-like is lower bounded by $(1 - \frac{M_{\ell^*}}{(\sparsity-M_{\ell^*})})^{M_{\ell+1} - M_{\ell}}$.

It now follows that the probability that ${\cal N}^{2\ell^*}_{\vec{e}}$ is tree-like is lower bounded by 
\[
\min_i \ \left(1 - \frac{M_{\ell^*}}{(\sparsity-M_{\ell^*})}\right)^{M_{\ell^*}} \left(1 - \frac{C_{\ell^*,i}}{\factor{i}-C_{\ell^*,i}}\right)^{C_{\ell^*,i}}
\]
Hence, for $\sparsity$ sufficiently large and fixed $\ell^*$,
\[
Pr\left( {\cal N}^{2\ell^*}_{\vec{e}} \text{not tree-like}\right) \leq \max_i \frac{M^2_{\ell^*}}{\sparsity} + \frac{C^2_{\ell^*,i}}{\factor{i}}
\]
Recall, the number of edges connected to a check node in set $i$ is a binomial $B(1/\factor{i},\sparsity)$ random variable, which asymptotically can be approximated by a Poisson random variable with mean $1/\binexp$, where $\factor{i} = \binexp \sparsity + O(1)$, for a constant $\binexp$. Hence, for any fixed value of $\ell^*$, $\max_i \{C_{\ell^*,i}\}$ and $M_{\ell^*}$ are upper bounded by $O((\log \sparsity)^{\ell^{*}})$ with probability $1/\sparsity^{c}$ for any constant $c$. Thus,
\[
Pr\left( {\cal N}^{2\ell^*}_{\vec{e}} \text{not tree-like}\right) \leq O\left(\frac{(\log \sparsity)^{2\ell^{*}}}{\sparsity}\right)
\]

\end{appendix}
\bibliographystyle{IEEEtran}
\bibliography{thesis}
\end{document}